\documentclass[12pt]{article}
\RequirePackage[colorlinks,citecolor=blue,linkcolor=blue,urlcolor=blue,pagebackref]{hyperref}
\usepackage{autonum}
\usepackage[utf8]{inputenc}
\usepackage{amsmath}
\usepackage{natbib}
\usepackage{bm}
\usepackage{amsthm}
\usepackage{mathtools}
\usepackage{amssymb}
\usepackage{comment}
\usepackage{commath}
\usepackage{changepage}
\usepackage{tikz}
\usetikzlibrary{tikzmark, calc}
\usepackage{xr}
\usepackage{setspace}
\usepackage{multirow}
\usepackage{booktabs}
\usepackage{setspace,float}
\usepackage{footmisc}
\usepackage{bbm}
\usepackage{threeparttable}
\usepackage{charter}
\usepackage{xfrac}
\definecolor{Myblue}{HTML}{3a99d9}
\definecolor{Myred}{HTML}{C0392B}
\definecolor{Mybluegreen}{HTML}{347986}
\definecolor{Mygreen}{HTML}{35837b}
\usepackage{xcolor}
\usepackage{mathrsfs}
\usepackage[english]{babel}
\usepackage[ruled,vlined]{algorithm2e}
\usepackage{tikz}
\usepackage[title]{appendix}
\usepackage[font=small,labelfont=bf]{caption}
\usepackage{enumerate}
\usepackage{dirtytalk}
\usetikzlibrary{decorations.pathreplacing, arrows,positioning}
\renewcommand{\arraystretch}{1.5}
\newtheorem{theorem}{Theorem}[section]
\newtheorem{assumption}{Assumption}
\newtheorem{example}{Example}

\newtheorem{lemma}{Lemma}[section]
\newtheorem{proposition}{Proposition}[section]
\theoremstyle{definition}
\newtheorem{definition}{Definition}[section]
\newtheorem{remark}{Remark}[section]

\DefineFNsymbols{mySymbols}{{\ensuremath\dagger}{\ensuremath\ddagger}\S\P
   *{**}{\ensuremath{\dagger\dagger}}{\ensuremath{\ddagger\ddagger}}}
\setfnsymbol{mySymbols}

\DeclareMathOperator*{\argmax}{arg\,max}

\DeclareMathOperator*{\argsup}{arg\,sup}

\makeatletter
\newcommand*\rel@kern[1]{\kern#1\dimexpr\macc@kerna}
\newcommand*\widebar[1]{%
  \begingroup
  \def\mathaccent##1##2{%
    \rel@kern{0.8}%
    \overline{\rel@kern{-0.8}\macc@nucleus\rel@kern{0.2}}%
    \rel@kern{-0.2}%
  }%
  \macc@depth\@ne
  \let\math@bgroup\@empty \let\math@egroup\macc@set@skewchar
  \mathsurround\z@ \frozen@everymath{\mathgroup\macc@group\relax}%
  \macc@set@skewchar\relax
  \let\mathaccentV\macc@nested@a
  \macc@nested@a\relax111{#1}%
  \endgroup
}
\makeatother
\usepackage{xr}
\makeatletter
\newcommand*{\addFileDependency}[1]{
\typeout{(#1)}
\@addtofilelist{#1}
\IfFileExists{#1}{}{\typeout{No file #1.}}
}\makeatother

\usepackage{dirtytalk}
\usetikzlibrary{decorations.pathreplacing, arrows,positioning}
\newcounter{algorithm}

\DeclareCaptionLabelFormat{alglabel}{Algorithm~#2}
\usepackage[a4paper]{geometry}
\begin{document}

\title{Individualized Treatment Allocation in Sequential Network Games\thanks{We thank Isaiah Andrews, Kirill Borusyak, Andrew Chesher, Timothy Christensen, Ben Deaner, Aureo de Paula, Duarte Gonçalves, Sukjin Han, Hiroaki Kaido, Keisuke Hirano, Daniel Lewis, Michael Leung, Jonas Lieber, SangMok Lee, Jordan Norris, Jeff Rowley, Shuyang Sheng, Davide Viviano, and Andrei Zeleneev for beneficial comments. We also benefited from the comments of the participants in seminars at Brown University, Chicago Booth, CUHK Shenzhen, Doshisha University, LSE, Stanford, UCL, Western Ontario, WUSTL, Yale, and 2023 AMES, ESG, NASMES, and SEA conferences. 
We also thank Ziyu Jiang, Tian Xie, Yanziyi Zhang, and Yuanqi Zhang for their helpful discussions. Kitagawa gratefully acknowledges financial support from the NSF (grant number: FAIN 234361).}}

\author{Toru Kitagawa\thanks{Brown University. Email$\colon$ \href{mailto:toru_kitagawa@brown.edu}{toru\_kitagawa@brown.edu}}\: and Guanyi Wang\thanks{Washington University in St. Louis. Email$\colon$ \href{mailto:guanyi@wustl.edu}{guanyi@wustl.edu}}}
\maketitle
\begin{abstract}
Designing individualized allocation of treatments so as to maximize the equilibrium welfare of interacting agents has many policy-relevant applications. 
    Focusing on sequential decision games of interacting agents, this paper develops a method to obtain optimal treatment assignment rules that maximize a social welfare criterion by evaluating stationary distributions of outcomes. 
    Stationary distributions in sequential decision games are given by Gibbs distributions, which are difficult to optimize with respect to a treatment allocation due to analytical and computational complexity. 
    We apply a variational approximation to the stationary distribution and optimize the approximated equilibrium welfare with respect to treatment allocation using a greedy optimization algorithm. 
    We characterize the performance of the variational approximation, deriving a performance guarantee for the greedy optimization algorithm via a welfare regret bound.
    We implement our proposed method in simulation exercises and an empirical application using the Indian microfinance data \citep{banerjee2013diffusion}, and show it delivers significant welfare gains.
\end{abstract}



\section{Introduction}\label{sectionintro}
The question of how best to allocate treatment to units interacting in a network is relevant to many policy areas, including the provision of local public goods \citep{bramoulle2007public}, the diffusion of microfinance (\citet{banerjee2013diffusion}; and \citet{akbarpour2025just}), and vaccination (\citet{galeotti2013strategic}; and \citet{KITAGAWA2023109}). 
Obtaining an optimal individualized allocation, however, is often infeasible due to analytical and computational challenges. 
As a consequence, practical counterfactual policy analysis in the presence of network spillovers is limited to simulating and comparing outcome distributions or welfare values across a few benchmark candidate policies. 
This leaves the magnitude of the potential welfare gains of an optimal individualized assignment policy unknown. 

Focusing on a class of social network models in which interacting agents play sequential decision games (\citet{jackson2002evolution}; \citet{nakajima2007measuring}; \citet{mele2017structural}; and \citet{christakis2020empirical}), this paper develops a method to obtain optimal treatment assignment rules that maximize a social welfare criterion.
We consider an individualized allocation of binary treatments over agents who are heterogeneous in terms of their own observable characteristics, their network configurations, and their neighbors' observable characteristics. 
Each agent chooses a binary outcome so as to maximize their own utility.
This choice depends upon the agent's own characteristics and treatment as well as their neighbors' characteristics, treatments and choices.
The sequential decisions of randomly ordered agents induce a unique stationary distribution of choices \citep{nakajima2007measuring,mele2017structural}. 
We specify the planner's welfare criterion to be the mean of the aggregate outcomes (i.e., the sum of the means of binary outcomes over all agents in the network) at the stationary distribution that is associated with a given treatment allocation.
We aim to maximize the welfare evaluated at the stationary outcome distribution with respect to the individualized allocation of treatments. Consider, for example, targeted information provision in villages with the aim of increasing microfinance adoption, as discussed in \citet{banerjee2013diffusion}. By comparing the stationary distribution of adoption decisions among units after they receive treatments, we determine whom to target in the village to maximize the adoption rate subject to a capacity constraint.

There are analytical and computational challenges to solving the maximization problem for optimal targeting. 
First, fixing an allocation of treatments, the sequential decision games induce a Markov random field (MRF) and the stationary outcome distribution has a Gibbs distribution representation.
The analytical properties of the mean of the aggregate outcomes, however, are difficult to characterize. 
To approximate the joint distribution of outcomes, the literature on MRFs performs numerical methods such as Markov Chain Monte Carlo (MCMC) \citep{geman1984stochastic}. 
If the size of the network is moderate to large though, MCMC can be slow to converge. 
It is, therefore, practically infeasible to perform MCMC to evaluate the welfare at every candidate treatment assignment policy in brute force search for optimal treatment allocation. 
Second, obtaining an optimal individualized assignment is a combinatorial optimization problem with respect to a binary vector of length $N$ with at most $\kappa \geq 1$ number of ones, where $N$ is the number of agents in the network and $\kappa$ is the capacity for the treatment.
Brute force search yields the global optimum; however, it requires evaluating welfare under $\binom{N}{\kappa}$ different allocations.   

We tackle these challenges by proposing a novel method for approximately solving the combinatorial optimization problem for individualized assignment. Our approach builds on variational approximation, which approximates the stationary distribution of the outcomes by a more tractable parametric family of distributions that minimizes the Kullback-Leibler divergence with the target Gibbs distribution.
We then optimize the variationally approximated equilibrium welfare with respect to the assignment vector by a greedy algorithm, which assigns the treatment sequentially to the unit who generates the largest welfare gain given the previous assignments. The variational approximation step reduces the computational burden of running MCMC at each candidate policy. The greedy optimization further reduces the computational complexity of the combinatorial optimization, since it suffices to evaluate the equilibrium welfare only at the order of $N \times \kappa$ times.

In addition to the novel proposal of how to obtain targeting policy, this paper comes with a couple of new theoretical contributions. First, we provide a theoretical justification for maximizing the welfare criterion under a variational approximation by showing that the associated approximation gap of welfare shrinks to zero as the size of network increases. To show this claim, we derive and exploit a novel transportation inequality that bounds the welfare regret by the Kullback-Leibler divergence, and justifies the use of variational approximation to search for welfare-optimal treatment allocation.
Second, we derive a theoretical regret bound to gauge the welfare performance of our greedy algorithm compared with the computationally infeasible brute force optimization.

To highlight this paper's unique contributions, we abstract from estimation of the structural parameters underlying the sequential decision game and assume that they are known. See \citet{geyer1992constrained}, \citet{snijders2002markov}, \citet{wainwright2008graphical}, \citet{chatterjee2013estimating}, \citet{mele2017structural}, \citet{boucher2017my}, and \citet{10.1162/rest_a_01023} for identification and estimation of these parameters, and \cite{wang2024robust} for how to incorporate estimation uncertainty into welfare regret bounds.
In practical terms, our proposed method is useful for computing an optimal assignment of treatment, with point estimates of the structural parameters plugged-in. 

To assess the performance of our proposal, we perform extensive numerical studies. In a small network setting where we can feasibly find a globally optimal policy by brute force, our simulations demonstrate that our greedy algorithm closely replicates a global optimum obtained by grid search.
We augment these numerical studies with an empirical application that illustrates the implementation and welfare performance of our method.
Specifically, we apply our procedure to Indian microfinance data that has been previously analyzed by \citet{banerjee2013diffusion}. 
The data contains information about households in a number of villages, their relation to other households in their village, and whether they chose to purchase a microfinance product.
For each village in the sample, we estimate the structural parameters of an assumed utility function using the method outlined in \citet{snijders2002markov}. 
Plugging in the parameter estimates, we obtain an individualized treatment allocation rule using our algorithm. 
We compare the village welfare attained by our proposed algorithm with the welfare achieved by the centrality-based allocation performed by an NGO called Bharatha Swamukti Samsthe (BSS). As new empirical findings to the literature, the welfare comparisons show that targeting central units in a network as done by NGO is suboptimal, implying that in addition to the network structure, the magnitudes of strategic interactions and treatment spillovers are important factors to consider when designing the treatment allocation.
For all $43$ villages in the sample, our method is associated with a higher welfare-level, with the magnitude of improvement varying from $9.82\%$ to $137.46\%$ (average improvement is $40.69\%$) of the welfare levels attained by the NGO's allocations.
The magnitude of the welfare gain is substantial and demonstrates the benefits of individualized targeting under interference. 


\subsection{Literature Review}
This paper intersects with several literatures in economics and econometrics, including graphical game analysis, MRF and variational approximation, discrete optimization of non-submodular functions, and statistical treatment rules.

Graphical game analysis has a long history in economics, see \citet{rosenthal1973class}, \citet{kakade2003correlated},
\citet{ballester2006s},
\citet{roughgarden2010algorithmic}, \citet{kearns2013graphical}, 
\citet{babichenko2016graphical}, \citet{de2018identifying}, \citet{leung2020equilibrium} and \citet{Parise2023graphon}. 
The most relevant paper to our work are \citet{mele2017structural} and \citet{christakis2020empirical}, which study strategic sequential network formation. \citet{mele2017structural} formulates the network formation game as a potential game \citep{monderer1996potential}, and characterizes the stationary distribution of the network as a Gibbs distribution. 
We apply a similar technique to generate the stationary distribution of actions in our game, while the main focus is to develop a method for approximating optimal targeting in terms of equilibrium welfare. 
\citet{kashaev2023peer} introduces a similar sequential structure into a discrete choice model with peer effect and estimates the model with panel data. 
\citet{badev2021nash} extends the setting in \citet{mele2017structural} to study how behavioral choices change the network formation. 
\citet{ballester2006s} and \citet{galeotti2020targeting} also study targeted interventions on networks, while the utility specification, the objective function, and the action space specified therein differ from ours. 

MRF offer a way to represent the joint distribution of random variables as a collection of conditional distributions. 
We model an individual's choice of outcomes as the maximization of a latent payoff function that depends upon a treatment allocation and their neighbors' choices, and derive an MRF representation of the joint distribution of outcomes. 
\citet{wainwright2008graphical} covers key results on variational approximation and MRF, and 
\citet{chatterjee2016nonlinear} provides an approximation error bound to variational approximation applied to MRF for binary outcomes. 
To the best of our knowledge, this literature has not studied how to obtain an optimal intervention in terms of a criterion function defined on the joint distribution of outcomes characterized as a MRF.


Although it does not introduce sampling uncertainty, this paper shares some motivation with the literature on statistical treatment rules \citep{manski2004statistical, dehejia2005program}.
In econometrics and machine learning, welfare regret typically arises due to uncertainty surrounding the value of underlying parameters (i.e., estimation).
In this work, regret arises from our use of variational approximation and of a greedy algorithm (i.e., identification). 
See \citet{stoye2009minimax,stoye2012minimax}, \citet{hirano2009asymptotics,hirano2020asymptotic}, \citet{chamberlain2011bayesian,chamberlain2020robust}, \citet{tetenov2012statistical}, and \citet{christensen2022optimal} for decision theoretic analyses of statistical treatment rules. 
There is also a growing literature on learning individualized treatment assignments including \citet{kitagawa2018should}, \citet{athey2021policy}, 
\citet{kasy2021adaptive}, \citet{kitagawa2021constrained}, \citet{mbakop2021model}, 
\citet{sun2021empirical}, and \citet{adjaho2022externally}, among others. 
These works do not consider settings that allow for the network spillovers of treatments or outcomes. 

There are some recent works that introduce network spillovers into statistical treatment choice.
\citet{viviano2025policy} and \citet{ananth2020optimal} consider treatment assignment rules taking into account the spillover effects of the treatments. 
In contrast to them, we consider spillovers through strategic interactions together with the treatment spillovers, and focus on fully individualized treatment assignment rules rather than a treatment rule as a constrained function of one's observable characteristics and network information.
\citet{munro2023treatment} studies targeting analysis taking into account spillovers through the market equilibrium. 
\citet{KITAGAWA2023109} considers the allocation of vaccines over an epidemiological network model (a Susceptible-Infected-Recovered network) in simple two-period transition model. \citet{wang2024robust} studies targeting over agents who play a simultaneous game with multiple equilibria.  
In a different context of treatment choice, \citet{kitagawa2022stochastic} applies variational approximation to a quasi-posterior distribution for individualized treatment assignment policies and studies welfare regret performances when assignment policies are drawn randomly from the variationally approximated posterior.

    \section{Model}\label{sectionmodel}
\subsection{Setup}

Let $\mathcal{N}=\{1,2,..., N\}$ be the set of individuals in a network. 
Each unit has a $k$-dimensional vector of observable characteristics that we denote by $X_i$, $i\in\mathcal{N}$. 
Assuming that the support of $X_i$ is bounded, we normalize the measurements of $X_i$ to be nonnegative, such that $X_i\in\mathbb{R}_{+}^{K}$. 
Let $\mathcal{X}=(X_{1},...,X_{N})$ be a matrix that collects the characteristics of units in the population, and let $\mathcal{X}^N$ denote the set of all possible matrices $\mathcal{X}$. Let $\overline{X}\coloneqq(\max_{i\in\mathcal{N}}X_{1i},...,\max_{i\in\mathcal{N}}X_{ki})$. 
Let $D=(d_1,...,d_N)$ denote a vector of treatment allocation, where $d_i\in\{0,1\}$, $i\in \mathcal{N}$, indicates whether unit $i$ is treated ($d_i=1$) or untreated ($d_i=0$).

The social network is represented by an $N\times N$ binary matrix that we denote by $G=\{G_{ij}\}_{i,j\in\mathcal{N}}$, and that is fixed and exogenous in this work. 
$G_{ij}=1$ indicates that units $i$ and $j$ are connected in the social network, whilst $G_{ij}=0$ indicates that they are not. 
Let $\mathcal{N}_i$ indicate the set of neighbors of unit $i$. 
$\widebar{N}$ denotes the maximum number of edges for one unit in the network (i.e., $\widebar{N}=\max_i\vert \mathcal{N}_i\vert$); 
$\underline{N}$ denotes the minimum number of edges for one unit in the network (i.e., $\underline{N}=\min_i\vert \mathcal{N}_i\vert$). 
As a convention, we assume there are no self-links (i.e., $G_{ii}=0,\:\forall i\in\mathcal{N}$). 
We further assume that the following property holds for the network structure$\colon$
\begin{assumption}{(\textbf{Undirected Link})}\label{assumpundir}
The adjacency matrix $G$ is symmetric.
\end{assumption}
\noindent The symmetric property of interaction in Assumption \ref{assumpundir} is a necessary condition for our interacted sequential decision game to be a proper potential game (Definition \ref{defpg} below) that can yield a unique stationary outcome distribution. 
The size of the spillover between units $i$ and $j$ depends not only upon $G_{ij}$ but also upon the treatment allocation and covariates, as we specify further below. 



Let $y = (y_1, \dots, y_N) \in \mathcal{Y}^N$, $\mathcal{Y}=\{0,1 \}$, be a profile of actions of the $N$ individuals in the network. Let the preferences (utilities) of individuals be given by $\{U_i(y,\mathcal{X},D,G;\boldsymbol{\theta})\}_{i=1}^N$, which depend on own and others' actions $y$, planner's treatment allocation $D$, observable characterstics $\mathcal{X}$, network structure $G$, and the structural parameters $\boldsymbol{\theta}$. We specify the individual utility function as a quadratic function of own and neighbors' choices:
    \begin{equation}\label{eq:uti}
    U_i(y,\mathcal{X},D,G;\alpha,\beta,\gamma)=\alpha_iy_i+\sum_{j\in\mathcal{N}_i}\beta_{ij}y_iy_j+\sum_{j\in\mathcal{N}_i}\gamma_{ij}y_j,
\end{equation}
where we let coefficients ($\alpha_i,\beta_{ij}, \gamma_{ij})$ depend on their own covariates and treatment status as well as those of their neighbor unit $j$.
Given a network $G$, covariates $\mathcal{X}$, and a treatment allocation $D$, the coefficient $\alpha_i$ on unit $i$'s choice represents the change in utility with respect to own action, and $\beta_{ij}+ \gamma_{ij}$ and  $\gamma_{ij}$ represent the spillovers of unit $j$'s action to $i$'s utility  when $y_i=1$ and $y_i=0$, respectively.  

In our empirical application, we specify
\begin{equation}\label{eq:alpha}
    \alpha_i=\theta_0+\theta_1d_i+X_i'(\theta_2+\theta_3d_i
)+A_N\theta_4 \sum_{j\in\mathcal{N}_i}m_{ij}d_j,
\end{equation}
and
\begin{equation}\label{eq:beta}
    \beta_{ij}=A_Nm_{ij}(\theta_5+\theta_6d_id_j).
\end{equation}
where $m_{ij}=m(X_i,X_j)$ is a (bounded) nonnegative real-valued function of personal characteristics measuring the distance between the characteristics of units $i$ and $j$. 
As shown in Proposition \ref{propoten} and Theorem \ref{theorem1} below, independent of their functional form specifications, $\gamma_{ij}$'s are not identifiable parameters since they do not appear in the stationary outcome distribution.
$A_N$ is a term that controls the magnitude of spillovers to make the unknown parameters ($\theta_4,\theta_5,\theta_6$) independent of the size of the network;
$A_N = \mathcal{O}(1/N)$ corresponds to a dense network and $A_N = \mathcal{O}(1)$ corresponds to a sparse network. Our framework does not restrict the density of network as far as the strategic interaction term, $\sum_{j\in\mathcal{N}_i}A_Nm_{ij}(\theta_5+\theta_6d_id_j)$, is bounded as network size increase.

\subsubsection*{\textbf{Running example: Information Injection on Microfinance \citep{banerjee2013diffusion}}} A network corresponds to a village consisting of $N$ households connected by a friendship network \(G\). Each household \(i\) makes a microfinance adoption decision \(Y_i\). Treatment $d_i$ indicates whether household \(i\) attends an information seminar on available microfinance products. Household $i$'s adoption decision could depend on whether their friends had attended the information seminar due to the information diffusion through friendship (i.e., spillover effects). Moreover, their adoption decision can depend on whether their friends adopt or not. e.g., households whose friends become wealthy due to microfinance are attracted to microfinance more than those households whose friends do not adopt microfinance. The NGO observes: the household characteristics \(X\), and social network $G$, and chooses a treatment allocation \(D\) to maximize the overall adoption rate in the village.

\begin{example}\textbf{(Customer Purchase Decisions)}
    Individual $i$ makes a purchase decision $Y_i$ (i.e., buy or not buy) for one product (e.g., Dropbox subscription, Orange from Sainsbury, iPhone). 
    In this example, the social planner is the company that is trying to maximize the total number of customers that purchase its products. 
    Individuals' purchase decisions sequentially depend upon the purchase decision of their friends or of their colleagues. 
    The company observes individuals' friendships and then decides how to allocate discount offers to achieve its own targets. (e.g., \citealp{richardson2002mining})
\end{example}

\begin{example}{(\textbf{Criminal Network})}
In a criminal network, suspects are connected by a social network. 
Suspect $i$ makes a decision whether to commit a crime, $Y_i=0$, or not, $Y_i=1$. 
The social planner in this example is the government or a police force that is trying to minimize the total number of crimes in the long run. 
The decision that a suspect makes about whether to commit a crime is based upon whether they and their friends have been arrested before ($d_i=1$ denotes they have been arrested before and $d_i=0$ denotes they have not been arrested in the past). 
The social planner observes the criminal network and decides which suspects to arrest. (e.g, \citealp{lee2021key})
\end{example}

\subsection{Potential Game}\label{sec:poten}
Before introducing the sequential structure of the game and deriving the stationary outcome distribution, we first formulate potential games. 
The concept of a potential game has been used to study strategic interaction since \citet{rosenthal1973class}, providing a tool to analyze the Nash equilibria of (non) cooperative games in various settings (e.g., \citealp{jackson2010social}, and \citealp{bramoulle2014strategic}).

\begin{definition}{(\textbf{Potential Game} \citep{monderer1996potential})}\label{defpg} Let $(U_i(y_i,y_{-i}): i=1, \dots, N)$ be the pay-offs of a game with $N$ players. It is a potential game if there exists a potential function $\Phi:\mathcal{Y}^N\rightarrow\mathbb{R}$ such that for all $i\in\mathcal{N}$ and for all $y_i,y_i'\in\mathcal{Y}$, it holds 
\begin{equation}\label{eqpt}
U_i(y_i,y_{-i})-U_i(y_i',y_{-i})=\Phi(y_i,y_{-i})-\Phi(y_i',y_{-i}).
\end{equation}
    
\end{definition}
The change in potentials from any player's unilateral deviation matches the change in their payoffs. 
Nash equilibria of the simultaneous move game with the pay-offs $(U_i(y_i,y_{-i}): i=1, \dots, N)$ must be the local maximizers of the potential. 
\citet[Theorem 4.5]{monderer1996potential} states that 
\begin{equation}\label{Eq:symmetric_utility}
\frac{\partial U_i}{\partial y_i\partial y_j}=\frac{\partial U_j}{\partial y_j\partial y_i}
\end{equation}
is a necessary and sufficient condition for a game featuring a twice continuously differentiable utility function to be a potential game.
For the discrete outcome case, a condition\footnote{%
Replacing the second-order derivative in Eq.\ref{Eq:symmetric_utility} with second-order differences. 
See \citet[\S Corollary 2.9]{monderer1996potential} for further details.} that is analogous to Eq.\ref{Eq:symmetric_utility}, is a necessary and sufficient condition for the existence of a potential function. 
\citet{chandrasekhar2014tractable}, and \citet{mele2017structural} also use a potential game framework to analyze Nash equilibria in a network game.  
This condition, however, requires that $\beta_{ij}=\beta_{ji}$ for all $i \neq j\in\mathcal{N}$, \footnote{
For a potential function to exist, after eliminating zero terms, we require that $U_i(1,0,y_{-ij})-U_i(1,1,y_{-ij})+U_j(1,1,y_{-ij})-U_j(0,1,y_{-ij})=0$. 
This implies that $-\beta_{ij}+\beta_{ji}=0$.} 
which restricts the spillover effect of unit $i$'s choice on unit $j$. 
By assuming our game is a potential game, we ensure the existence of at least one pure strategy Nash equilibrium \citep{monderer1996potential}. The next proposition shows a potential function for the utility function specified in Eq. \ref{eq:uti}.

\begin{proposition}{(\textbf{Potential Function})}\label{propoten}
Under Assumption \ref{assumpundir}, a potential function $\Phi(y,\mathcal{X},D,G;\boldsymbol{\theta})$ for $U_i(y,\mathcal{X},D,G;\boldsymbol{\theta})$ specified in Eq.\ref{eq:uti} can be defined as:
\begin{equation}\label{eqpotentialgame}
\begin{split}
\Phi(y,\mathcal{X},D,G;\boldsymbol{\theta})&=\sum_{i=1}^N\alpha_iy_i+\frac{1}{2}\sum_{i=1}^N \sum_{j\in\mathcal{N}_i} \beta_{ij} y_iy_j.
        \end{split}
    \end{equation} 
\end{proposition}
Proof of Proposition \ref{propoten} is provided in Appendix \ref{appenlemma1}.
Notice that the potential function is not the summation of the utility function across all units;
summation of the utility function counts the interaction terms twice and violates Eq.\ref{eqpt}. As we show in Theorem \ref{theorem1} below, the likelihood depends on the structural parameters only through the potential function, implying that $\gamma_{ij}$ cannot be identified by the data. As far as identification of the other parameters $(\alpha_i, \beta_{ij})$ is concerned, we can normalize $\gamma_{ij} = 0$ without loss of generality, while $\gamma_{ij} = 0$ is not an innocuous normalization if one is interested in counterfactual analysis on agents' utilities. This is because in the presence of the spillovers through actions, utility level of $i$'s baseline choice $y_i = 0$ depends on the neighbors' actions and setting $\gamma_{ij} = 0$ shuts down this channel of spillovers.

By characterizing our game as a potential game, we can employ the stationary outcome distribution that we derived in Theorem \ref{theorem1} to evaluate the planner's expected welfare.

\subsubsection*{\textbf{Running example: Information Injection on Microfinance \citep{banerjee2013diffusion}}}  To guarantee the existence of a potential function, our analysis imposes the structure of bilateral spillovers and their symmetries. The bilateral structure of the spillovers is a common specification in the peer effect model, e.g., \cite{nakajima2007measuring}. The symmetry of the magnitude of the spillovers of household $i$, $\beta_{ij} = \beta_{ji}$, can restrict heterogeneity of units. For instance, in the context of microfinance program with the source of spillovers being consumption externality through comparisons with friends, the symmetry assumption holds if friends share unobserved heterogeneity such as prone to jealousy.

\subsection{Sequential Decision Process}\label{sectionsdp}

We now introduce sequential games and corresponding stationary outcome distributions. Let $t \in \{ 1,2, \dots \}$ index a period and $Y_i^t\in\mathcal{Y}$ be unit $i$'s choice (outcome) made at time $t$. 
Let $Y^t = (Y_1^t,...,Y_N^t ) \in \mathcal{Y}^N$ and we view $\{ Y^t: t = 1,2,\dots \}$ be a stochastic process with its realization denoted by $\{ y^t = (y_i^t)_{i=1}^N : t=1,2, \dots \}$. 
We denote the outcome vector excluding $y^t_i$ by $y_{-i}^t$.

In the initial period $t=0$, the social planner observes the connections in the social network and individuals' attributes, and decides the treatment allocation $D$. 
Then, at the beginning of every period $t=1,2,\dots$, an individual $i$ is randomly chosen from $\mathcal{N}$ by nature.
The chosen unit $i$ can set own action (outcome) $y_i^t$ given the action profiles of the others in the previous period $y_{-i}^{t-1}$. 
All the other units maintain the same choices as in the last period. 

Let $O^t \in \mathcal{N}$ indicate a unit chosen by nature in period $t$. The next assumption restricts the dependence of $\{O^t: t=1,2,\dots\}$ on other variables in the model.  


\begin{assumption}{(\textbf{Decision Process})}\label{assdp}
The probability of unit $i$ being selected at time $t$ given $(y^{t-1},\mathcal{X},D,G)$ is positive and does not depend upon $y_{i}^{t-1}$ for all $i \in \mathcal{N}$ and $t = 1,2,\dots$, i.e.,
    \begin{equation}
        \rho_i(y_{-i}^{t-1}) := \Pr(O^t=i\vert y^{t-1},\mathcal{X},D,G) = \Pr(O^t=i\vert y_{-i}^{t-1},\mathcal{X},D,G) >0.
    \end{equation}
\end{assumption}

We assume that the selected unit $i$ in period $t$ chooses own action $y_i^t$ so as to maximize their current utility myopically, observing the attributes, treatment status, and the actions of their neighbors.
Before doing so, however, the selected unit $i$ receives idiosyncratic shocks $(\varepsilon_{1it}, \varepsilon_{0it})$. 
Unit $i$ then chooses $Y_i^t=1$ if and only if:
\begin{equation}\label{eq:us}
    U_i(1,y_{-i}^{t-1},\mathcal{X},D,G;\boldsymbol{\theta})+\varepsilon_{1it}\geq U_i(0,y_{-i}^{t-1},\mathcal{X},D,G;\boldsymbol{\theta})+\varepsilon_{0it}.
\end{equation}
The next assumption specifies the distribution of the idiosyncratic shocks.
\begin{assumption}{(\textbf{Utility Shocks})}\label{assps}
    $\varepsilon_{1it}$ and $\varepsilon_{0it}$ follow the Type 1 extreme value distribution and are independent and identically distributed among units and across time.
\end{assumption}

Note that the Nash equilibrium of the potential game in Section \ref{sec:poten} is defined with the utilities \textit{without} the utility shocks. Hence, the shocks $(\varepsilon_{1it},\varepsilon_{0it})$ added in the sequential setting should be interpreted as optimization errors of the selected individual, rather than her unobserved characteristics.

Under Assumption \ref{assps}, the conditional probability of unit $i$ choosing $Y_i^t=1$ is:
\begin{equation}
    P(Y_{i}^t=1\vert Y_{-i}^{t-1}=y_{-i}^{t-1},\mathcal{X},D,G;\boldsymbol{\theta})=\frac{\exp[U_i(1,y_{-i}^{t-1},\mathcal{X},D,G;\boldsymbol{\theta})]}{\sum_{y_i\in\{0,1\}}\exp[U_i(y_i,y_{-i}^{t-1},\mathcal{X},D,G;\boldsymbol{\theta})]}.
\end{equation}
Given initial value $y^0$, sequence $[Y^1,...,Y^t, \dots]$ evolves as a Markov chain such that:
\begin{equation}
            Y^t_i|(Y^{t-1}=y^{t-1}) = \begin{cases} y^{t-1}_i & \text{w/p} \mspace{10mu}  1-\rho_i(y_{-i}^{t-1}) \\
            y & \text{w/p} \mspace{10mu} \rho_i(y_{-i}^{t-1}) \cdot P(Y_{i}^t=y\vert Y_{-i}^{t-1}=y_{-i}^{t-1}),
            \end{cases} 
             \quad \forall i\in\mathcal{N},
\end{equation}
for $y\in\{0,1\}$.
Under Assumptions \ref{assumpundir} to \ref{assps}, this Markov chain is \textit{irreducible} and \textit{aperiodic},\footnote{%
It is irreducible since every configuration could happen in a finite time given our assumption on the selection process. 
It is aperiodic since the selected unit has a positive probability to choose the same choice as in the last period.} 
which has a unique stationary distribution. 

The individual decision process is a stochastic best response dynamic process \citep{blume1993statistical}, which evolves as a Markov chain of decisions. 
\citet{jackson2002evolution} shows that the sequential decision process plays the role of a stochastic equilibrium selection mechanism in the static potential game of Section \ref{sec:poten}. 
\citet{lee2009multiple} performs counterfactual predictions of policy interventions in the presence of multiple equilibria, with best response dynamics playing the role of an equilibrium selection mechanism.

\subsection{Stationary Distribution}
Given the Markov chain structure of $\{ Y^t \}$, the stationary joint distribution of the outcomes in our sequential decision game is given by the next theorem. We denote the random variable following the stationary oucome distribution by $Y \in \mathcal{Y}$.

\begin{theorem}{\textbf{Unique Stationary Distribution} \citep{nakajima2007measuring,mele2017structural}:}\label{theorem1} Under Assumption \ref{assumpundir} to \ref{assps}, the outcomes of the sequential game $\{Y^t \}$ has a unique stationary distribution:
    \begin{equation}\label{eqjointdis}
             P[Y=y\vert \mathcal{X},D,G;\boldsymbol{\theta}]=\frac{\exp[\Phi(y,\mathcal{X},D,G;\boldsymbol{\theta})]}{\sum_{\delta\in\{0,1\}^N} \exp[\Phi(\delta,\mathcal{X},D,G;\boldsymbol{\theta})]}.
         \end{equation}
    \end{theorem}

Theorem \ref{theorem1} shows that, given the parametric specification of the distribution of shocks (Assumption \ref{assps}), the joint distribution of the outcomes is given by a Gibbs distribution characterized by the potentials.\footnote{\citet{hsieh2020specification} obtains the same Gibbs distribution by introducing configuration specific shocks.\label{foothsieh}} The set of Nash equilibria of the potential game corresponds to the set of modes of the stationary distribution. 

We can view the joint distribution of outcomes in the stationary distribution as a Markov random field (see, e.g., \citealp{bremaud2013markov}). 
The random field $\{Y_i\}_{i=1}^N$ is a collection of random variables on the state space $\mathcal{Y}$. 
This random field is a Markov random field if for all $i\in\mathcal{N}$ and $y\in\mathcal{Y}^N$:
\begin{equation}\label{eq:mrfcon}
    P(Y_i=y_i\vert Y_{-i}=y_{-i})=P(Y_i=y_i\vert Y_{j\in \mathcal{N}_i}=y_{j\in \mathcal{N}_i}),
\end{equation}
where $Y_{j \in \mathcal{N}_i}$ is the subvector of $Y$ restricted to $i$'s neighbors.

Given the specification of our utility function, the conditional distribution of $Y_i$ satisfies this Markov property. 
By connecting $Y$ to MRF, the \textit{Hammersley-Clifford Theorem} \citep{clifford1971markov, besag1974spatial} establishes that the joint distribution of $Y$ must follow a \textit{Gibbs distribution}, which is consistent with the result of Theorem \ref{theorem1}.

The stationary distribution of the outcomes shown in Theorem \ref{theorem1} is structural in the sense that the specification of the potential function in the Gibbs distribution relies on the functional form specification of the latent payoff function of agents. 
An advantage of the current structural approach is that we are transparent about the assumptions that we impose on the behavior of agents, on the structure of social interaction, and on the equilibrium concept. 
The structural approach, accordingly, disciplines the class of joint distributions of observed outcomes to be analyzed. 
As an alternative to the structural approach, we can consider a reduced-form approach where we model the conditional distribution of the observed outcomes given the treatment vector. 
Maintaining the family of Gibbs distributions, the reduced-form approach corresponds to introducing a more flexible functional form for the potential functions without guaranteeing that it is supported as a Nash equilibrium of the potential game. 

\subsubsection*{\textbf{Running example: Information Injection on Microfinance \citep{banerjee2013diffusion}}} The structure of sequential game introduced above can model the communication process among the villagers, e.g., in each short time window $t$,  a villager can revise her/his decision of taking up the treatment by knowing the decisions of the friends. 
After many periods of interactions, the framework leads to the joint distribution of households' take-up decisions as presented in Theorem \ref{theorem1}. 
In practical terms, we view a snapshot of the take-up status $Y$ in a given period after treatment assignment as a draw from the stationary Gibbs distribution of dimension $N$. The sequential structure is a theoretical device to obtain the stationary outcome distribution, and we do not require the training sample for estimating the structural parameters to be panel data with observations of $\{ Y^t \}$.

\section{Treatment Allocation}\label{sectionall}
The objective of the social planner is to select a treatment assignment $D^* \in \{0, 1 \}^N$ that maximizes equilibrium mean outcomes subject to a capacity constraint that the number of individuals that are treated cannot exceed $\kappa > 0$:

\begin{equation}\label{eqobj}
    D^*\in\argmax_{D\in\{0,1\}^N}\sum_{i=1}^N\mathbb{E}_P[Y_i\vert \mathcal{X},D,G;\boldsymbol{\theta}],
    \end{equation}
    
\begin{equation}
    s.t.\quad\sum_{i=1}^N d_i\leq \kappa,
\end{equation}
where $\mathbb{E}_P[Y_i\vert \mathcal{X},D,G;\boldsymbol{\theta}]$ is the expectation with respect to the stationary joint distribution shown in Theorem \ref{theorem1}. 


In this work, we assume that the structural parameters
$\boldsymbol{\theta}$ underlying the potential game are given and abstract from uncertainty in parameters estimation. 
There have been several estimation approaches studied in the literature which utilize MCMC \citep{geyer1992constrained,snijders2002markov}, pseudo-likelihood \citep{besag1974spatial,boucher2017my}, and variational approximations \citep{wainwright2008graphical}. 

Our proposal fits to the following scenario of empirical policy design. First, a planner estimates the structural parameters using a training sample of networks that have implemented the treatments. Next, using our proposed method, the planner informs optimal allocations of treatments to the networks which have not implemented the treatments, assuming that the structural parameters are invariant among them.

\subsection{Welfare Approximation}
We cannot directly maximize the equilibrium welfare in $D$;
instead, we seek to maximize the approximated welfare. 
We now discuss what prevents us from maximizing the equilibrium welfare.

The objective function $W(D)$ from Eq.\ref{eqobj} is:
\begin{equation}\label{eqz}
    \begin{split}
        W(D)&=\sum_{i=1}^N\mathbb{E}_P[Y_i\vert \mathcal{X},D,G]=\sum_{i=1}^N \sum_{y\in\{0,1\}^N}y_i\frac{\exp[\Phi(y,\mathcal{X},D,G;\boldsymbol{\theta})]}{\sum_{\delta\in\{0,1\}^N} \exp[\Phi(\delta,\mathcal{X},D,G;\boldsymbol{\theta})]}.
    \end{split}
\end{equation}
We define the \textit{partition function} as $\mathcal{Z}\coloneqq \sum_{\delta\in\{0,1\}^N} \exp[\Phi(\delta,\mathcal{X},D,G;\boldsymbol{\theta})]$.
Since the partition function $\mathcal{Z}$ sums all possible configurations (of which there are $2^N$), it is infeasible to evaluate the expectation when $N$ is moderate to large, e.g., when $N> 276$, there are more configurations than atoms in the observable universe \citep{de2020econometric}.

Given this well-known problem, we seek to approximate the distribution $P$ using a tractable distribution $Q$.
Let $\mu_i^P\coloneqq\mathbb{E}_{P}[Y_i\vert \mathcal{X},D,G]$ and $\mu_i^Q\coloneqq\mathbb{E}_{Q}[Y_i\vert \mathcal{X},D,G]$. Since we have:
\begin{equation} \label{eq:Pinsker}
        \begin{split}
            W(D)&= \sum_{i=1}^N \mu_i^{P}\leq \sum_{i=1}^N \vert\mu_i^{P}-\mu_i^{Q}\vert+\sum_{i=1}^N\mu_i^{Q},
        \end{split}
    \end{equation}  
The approximation error can be bounded by:
\begin{equation}\label{equpp}
    \sum_{i=1}^N \mu_i^{P}-\sum_{i=1}^N \mu_i^{Q}\leq\sum_{i=1}^N \vert\mu_i^{P}-\mu_i^{Q}\vert.
\end{equation}
In what follows, we bound the approximation error in Eq.~\ref{equpp} first by the Wasserstein distance (Lemma \ref{lemma:eb}) and next by the Kullback–Leibler (KL) divergence $\mathbb{KL}(Q\Vert P)$ (Theorem \ref{thm:klthm}). Specifically, we consider the Wasserstein 1-distance based on Hamming distance, as defined below

\begin{definition}[Wasserstein
distance $W_1$]
    Let $P$ and $Q$ be two probability distributions over $\{0,1\}^N$.  
Define $\Omega(P,Q)$ as the set of all couplings of $P$ and $Q$, i.e., all joint distributions of $(Y, Y')$ such that $Y \sim P$ and $Y' \sim Q$. Let the Hamming distance be
$d_{\mathrm{H}}(Y, Y') \;=\; \sum_{i=1}^N \mathbbm{1}\{Y_i \neq Y_i'\}.$
Then the Wasserstein 1-distance equipped with Hamming distance is given by
$$
W_1(P,Q)
\;=\;
\inf_{\omega \in \Omega(P,Q)}
\mathbb{E}_\omega\bigl[d_{\mathrm{H}}(Y, Y')\bigr].$$
\end{definition}
\begin{lemma}[\textbf{Error Bound}]\label{lemma:eb}
    Let $P$ and $Q$ be two probability distributions over $\{0,1\}^N$. We have:
    \begin{equation}
     \sum_{i=1}^N \vert\mu_i^{P}-\mu_i^{Q}\vert\leq W_1(P,Q).
    \end{equation}
\end{lemma} 
Proof of Lemma \ref{lemma:eb} is provided in Appendix \ref{lemmaproof:eb}. To obtain an upper bound of the Wasserstein 1-distance in terms of $\mathbb{KL}(Q\Vert P)$, we develop a version of Talagrand's transportation inequality for Gibbs distributions.\footnote{A distribution $P$ on $\mathbb{R}^N$ satisfies Talagrand's transportation inequality with constant $C>0$ if for any probability measure 
$Q$ absolutely continuous w.r.t. $P$, we have 
\refstepcounter{equation}\label{eq:talagrand} $$W_2(P,Q)\leq C\sqrt{\mathbb{KL}(Q\|P)}, \hspace{1cm}(\theequation)$$ where $W_2(P,Q)= \inf_{\omega \in \Omega(P,Q)}
\mathbb{E}_\omega\bigl[d^2_{\mathrm{H}}(Y, Y')\bigr]$
. Since $d_H(Y,Y')\leq d_H^2(Y,Y')$ holds for any $Y,Y' \in \{0,1\}^N$, Talagrand's transportation inequality 
implies $W_1(P,Q)\leq W_2(P,Q)\leq C\sqrt{\mathbb{KL}(Q\Vert P)}$ for any $Q$ absolutely continuous w.r.t. $P$.}
Availability of Talagrand's inequality hinges on the following assumption called as Dobrushin’s condition \citep{dobrushin1970prescribing}, which restricts the magnitude of the spillover effects. This Dobrushin’s condition also appears in the incomplete information network games to guarantee the uniqueness of the pure strategy Bayesian Nash equilibrium \citep{lee2014binary}.\label{footdobur}
\vspace{-0.2cm}
\begin{assumption}{\textbf{Limited Interactions}:}\label{ass:boundspill} The coefficients of strategic interaction satisfy
    \begin{equation}
        \max_{i=1,...,N}\sum_{j\in\mathcal{N}_i}\beta_{ij} < 4.
    \end{equation}
\end{assumption}
With the utility specification of Eq.\ref{eq:alpha} and \ref{eq:beta}, this condition can be rewritten as:
\begin{equation}
        \vert\theta_5\vert+\vert\theta_6\vert < 4\big(A_N\max_{i=1,...,N}\sum_{j=1}^Nm_{ij}G_{ij}\big)^{-1}. 
    \end{equation}With Dobrushin's condition added, we obtain the following key lemma:

\begin{lemma}{\textbf{$W_1$ Transportation Inequality for Binary Gibbs distributions}}\label{pro:tala}:
    Under Assumption \ref{assumpundir} to \ref{ass:boundspill}, there exists a universal constant $C_{trans}>0$ such that the stationary distribution $P$, defined in Eq.\ref{eqjointdis}, satisfies:
    \begin{equation} \label{eq:W1_transportation}
        W_1(P,Q)\leq C_{trans}\sqrt{\mathbb{KL}(Q\Vert P)},
    \end{equation}
    for all probability measures $Q$ on $\{0,1\}^N$.
\end{lemma}
Proof of Lemma \ref{pro:tala} is provided in Appendix \ref{lemmaproof:tI}.\footnote{\label{footnote:univ constant} The proof shows existence of a universal constant $C_{trans} > 0$ while it does not provide an explicit formula for the constant. We leave a determination of a universal constant for future research.} To our knowledge, the trasportation inequality in this lemma is new and can be of independent interest. Combining Lemmas \ref{lemma:eb} and \ref{pro:tala}, we have the following theorem that bounds the welfare approximation error by the KL divergence.
\begin{theorem}\label{thm:klthm}
    Under Assumption \ref{assumpundir} to \ref{ass:boundspill}, there is a universal constant $C_{trans}$ such that the stationary distribution $P$, defined in Eq.\ref{eqjointdis}, satisfies:
    \begin{equation}
\sum_{i=1}^N \vert\mu_i^{P}-\mu_i^{Q}\vert \leq C_{trans}\sqrt{\mathbb{KL}(Q\Vert P)}.
\end{equation}
for all probability measures $Q$ on $\{0,1\}^N$.
\end{theorem}
This novel transportation
inequality links the approximation error of the welfare to the Kullback–Leibler divergence. A similar inequality holds for a class of social welfare functions
as long as their approximation errors can be bounded by the Wasserstein distance.
Given such bound, in the next section, we obtain a best approximation of $P$ by $Q$ by minimizing $\mathbb{KL}(Q\Vert P)$ with respect to $Q$ subject to that $Q$ belongs to a simple class of distributions of $Y$. 
The reason that we focus on minimizing the KL divergence rather than the Wasserstein distance is because minimizing the KL divergence is computationally more tractable when $P$ is a Gibbs distribution. This computational advantage is well known and exploited in the literature of variational approximation for exponential random graphs \citep{wainwright2008graphical}; see Section~\ref{sectionmfm} below.

\begin{remark}\label{remark:31}
We may want to target maximizing the expected utilitarian welfare (i.e., the summation of individual utilities) when choosing the optimal treatment allocation. In which case, the objective function becomes
\begin{align}
        W_U(D) &= \sum_{i=1}^N \mathbb{E}_P[U_i(y,\mathcal{X},D,G;\boldsymbol{\theta})\vert \mathcal{X},D,G] \\ &=\sum_{i=1}^N \left[\alpha_i + \sum_{j \neq i}^N \beta_{ij} \mathbb{E}_P[y_i y_j\vert \mathcal{X},D,G] + \sum_{j \neq i}^N \gamma_{ij} \mathbb{E}_P[y_j\vert \mathcal{X},D,G] \right] 
\end{align}
As discussed in Section \ref{sec:poten}, however, when agent $i$'s utility level with $y_i=0$ depends on others' actions, i.e., $\gamma_{ij} \neq 0$ in Eq. \ref{eq:uti}, we cannot identify the utility level so that maximizing $W_U(D)$ empirically is not feasible. Note also that setting $\gamma_{ij} = 0$ is not an innocuous normalization since it alters the objective function for $D$. This contrasts with the welfare criterion $W(D)$ of (Eq. \ref{eqobj}) that does not depend on $\gamma_{ij}$.
\end{remark}

\subsection{Mean Field Method}\label{sectionmfm}
Using an independent Bernoulli distribution to approximate the target distribution is called \textit{naive mean field approximation} \citep{wainwright2008graphical}. 
This method can be viewed as a specific method in the general approach of \textit{variational approximation}, which approximates a complicated probability distribution by a distribution belonging to a class of analytically tractable parametric distributions.
In Eq.\ref{eq:Pinsker}, $P$ corresponds to the target distribution to be approximated and $Q$ corresponds to a simple parametric distribution approximating $P$.
We consider the class of independent Bernoulli distributions as a parametric family for $Q$, since it delivers a feasible and fast optimization algorithm and the magnitude of its approximation error is already established in the literature. 

The probability mass function of an independent Bernoulli distribution $Q$ is expressed as:

\begin{equation}
    Q(Y=y)=\prod_{i=1}^N (\mu_i^Q)^{y_i}(1-\mu_i^Q)^{1-y_i}.
\end{equation}

Let $\mu^Q$ be an $N\times 1$ vector that collects $\{\mu_i^Q\}_{i=1}^N$. 
The Kullback–Leibler divergence between $Q$ and $P$ equals:

\begin{equation}
       \begin{split}
            \mathbb{KL}(Q\Vert P) &= \mathbb{E}_Q\Big[\log\frac{Q(y)}{P(y)}\Big]\\
            &=\log \mathcal{Z}-\Big[\alpha'\mu^Q+\frac{1}{2}(\mu^{Q})'\mathcal{B}\mu^Q-\sum_{i=1}^N\big[\mu_i^Q\log(\mu_i^Q)+(1-\mu_i^Q)\log(1-\mu_i^Q)\big]\Big],
       \end{split}
    \end{equation}
    
where $\alpha\coloneqq(
\alpha_i)_{i\in\mathcal{N}}'$ is a $N\times 1$ weighting vector and $\mathcal{B}\coloneqq(\beta_{ij})_{i,j\in\mathcal{N}}$ is a $N\times N$ matrix. The last line holds since the diagonal entries of $\mathcal{B}$ are zero and 
\begin{equation}
    \begin{split}
        \mathbb{E}_Q[y'\mathcal{B} y]&=\mathbb{E}_Q\Big[\sum_{i=1}^N\sum_{j\neq i}^N\beta_{ij}y_iy_j\Big]=\sum_{i=1}^N\sum_{j\neq i}^N\beta_{ij}\mathbb{E}_Q[y_iy_j]=\sum_{i=1}^N \sum_{j\neq i}^N \beta_{ij}\mathbb{E}_Q[y_i]\mathbb{E}_Q[y_j]\\&=(\mu^{Q})'\mathcal{B}\mu^Q.
    \end{split}
\end{equation}
Recall $\mathcal{Z}$ in Eq.\ref{eqz} sums over all possible configurations. 
$\mathcal{Z}$ is, therefore, independent of $Y$ (i.e., it is constant). 
We define $\mathcal{A}(\mu^Q,\mathcal{X},D,G)$ as:
\begin{equation}\label{eq:A}
    \mathcal{A}(\mu^Q,\mathcal{X},D,G) \coloneqq \alpha'\mu^Q+\frac{1}{2}(\mu^{Q})'\mathcal{B}\mu^Q-\sum_{i=1}^N\big[\mu_i^Q\log(\mu_i^Q)+(1-\mu_i^Q)\log(1-\mu_i^Q)\big].
\end{equation}
As such, minimizing $\mathbb{KL}(Q\Vert P)$ is equivalent to maximizing $\mathcal{A}(\mu^Q,\mathcal{X},D,G)$ \citep{wainwright2008graphical}. 
We denote by $\Tilde{\mu}$ the result of the following optimization:
\begin{equation}\label{eqva}
    \begin{split}
        \Tilde{\mu}&=\argsup_{\mu^Q}\mathcal{A}(\mu^Q,\mathcal{X},D,G)\\&= \argsup_{\mu^Q}\alpha'\mu^Q+\frac{1}{2}(\mu^Q)'\mathcal{B}\mu^Q-\sum_{i=1}^N\big[\mu_i^Q\log(\mu_i^Q)+(1-\mu_i^Q)\log(1-\mu_i^Q)\big].
    \end{split}
\end{equation}
Then the approximated distribution $Q^*$ is expressed as:
\begin{equation}
    Q^*(Y=y)=\prod_{i=1}^N (\Tilde{\mu}_i)^{y_i}(1-\Tilde{\mu}_i)^{1-y_i}.
\end{equation}
The first order condition of Eq.\ref{eqva} is:
\small
\begin{equation}\label{eqmui}
\begin{split}
        \Tilde{\mu}_i &= \frac{1}{1+\exp[-(\alpha_i+\sum\limits_{j\in\mathcal{N}_i}\beta_{ij}\Tilde{\mu}_j)]}= \Lambda \Big[\alpha_i+\sum\limits_{j\in\mathcal{N}_i}\beta_{ij}\Tilde{\mu}_j\Big].
\end{split}
\end{equation}
\normalsize
Despite the fact that our approach speaks to a general specification of the utility function, we use the specification of Eq.\ref{eq:alpha} and \ref{eq:beta} to derive the theoretical results in the following sections for illustrative purposes. Although the objective function is non-concave, optimization of (Eq.\ref{eqva}) can make use of an iterative search of a fixed point, as proposed in the literature \citep{wainwright2008graphical}. We summarize it in the next algorithm:

\medskip

\refstepcounter{algorithm}\label{almu}
ALGORITHM 1:
\begin{enumerate}
\item[] \hspace{-0.5cm}Step 1. Initialize the iteration. For each unit $i\in\mathcal{N}$, draw an initial value $\tilde{\mu}_i^0 \sim U[0,1],$
and collect these draws into the vector $\tilde{\mu}^0=(\tilde{\mu}_1^0,\dots,\tilde{\mu}_N^0)'$. Set $t=1$.
\item[] \hspace{-0.5cm}Step 2. Given the previous iterate $\tilde{\mu}^{t-1}$, update the components of $\tilde{\mu}^t$ sequentially, i.e., for each $i=1,\dots,N$, compute $\tilde{\mu}_i^t \gets \Lambda\left(\alpha_i+\sum_{j\in\mathcal{N}_i}\beta_{ij}\tilde{\mu}_j^{\,t-1}\right).$
\item[] \hspace{-0.5cm}Step 3. After all units have been updated, evaluate the convergence criterion by computing the change in the objective function, $\Delta^t := \mathcal{A}(\tilde{\mu}^t,\mathcal{X},D,G)-\mathcal{A}(\tilde{\mu}^{t-1},\mathcal{X},D,G)$. Given a tolerance level $\varrho$, iterate Step 2 while $\Delta^t > \varrho$. If $\Delta^t \le \varrho$, then terminate the iteration and define the fixed-point as $\tilde{\mu}\gets \tilde{\mu}^t.$
\end{enumerate}
\medskip

Algorithm \ref{almu} shows an iterative procedure to compute variationally approximated mean value $\tilde{\mu}$. This algorithm performs $\mathcal{O}(N)$ operations in every iteration of Step 2 and, as shown in Proposition \ref{prounimax} below, this procedure is guaranteed to converge to the global optimum.  
As an alternative to variational approximation, MCMC can be used to simulate the mean value $\mu$ of the unique stationary distribution (Eq.\ref{eqjointdis}).
Since MCMC may require exponential time for convergence \citep{chatterjee2013estimating} though, simulating $\mu$ is infeasible for a large network (i.e., MCMC needs to be run $\mathcal{O}(\kappa N)$ times).
The next proposition shows that the optimization of (Eq.\ref{eqva}) has a unique maximizer and the iteration in  Algorithm \ref{almu} converges to it.

\begin{proposition}{\textbf{Unique Maximizer}:}\label{prounimax}
    Under Assumptions \ref{assumpundir} to \ref{ass:boundspill} and the utility function specification of Eq.\ref{eq:alpha} and \ref{eq:beta}, the optimization problem (Eq.\ref{eqva}) defining $\tilde{\mu}$ has a unique maximizer and the iteration procedure of Algorithm \ref{almu} converges to it. 
\end{proposition}
Proof of Proposition \ref{prounimax} is provided in Appendix \ref{appenum}. The iteration algorithm has been used in the literature \citep{wainwright2008graphical}, while to our knowledge, the conditions for the contraction mapping property and convergence to the global optimum shown in Proposition \ref{prounimax} are new. 

The following theorem shows how the approximation error due to variational approximation (measured in terms of the Kullback-Leibler divergence) depends on the size of the network.
\vspace{-0.3cm}
\begin{theorem}{\textbf{Approximation Error Bound}:}\label{theorem2}
        Let $Q^*$ denote the independent Bernoulli distribution solving Eq.\ref{eqva}. 
        Under Assumptions \ref{assumpundir} to \ref{assps} and the utility function specification of Eq.\ref{eq:alpha} and \ref{eq:beta}, the Kullback–Leibler divergence of $Q^*$ from $P$ is bounded from above by:
        \small
        \begin{equation}\label{eqkl}
            \mathbb{KL}(Q^*\Vert P)\leq C_1A_N\widebar{N}+C_2N+\mathcal{O}\left(\sqrt{A_N^2\widebar{N}^2N}\right)+\mathcal{O}\left(\sqrt{A_N^3\widebar{N}^2N^2}\right)+o(N),
        \end{equation}
        \normalsize
        where $C_1,C_2$ are known constants that depend only upon $\boldsymbol{\theta}$ and $\overline{m}\coloneqq\max_{i,j}m_{ij}$, and $A_N$ is defined in Eq.\ref{eq:alpha}.
    \end{theorem}
This theorem follows as a corollary of \citet[\S Theorem 1.6]{chatterjee2016nonlinear}. Proof of Theorem \ref{theorem2} is provided in Appendix \ref{apptheo31}.
Theorem \ref{theorem2} shows that the upper bound on the approximation error depends not only on the size of the network $N$ but also on the magnitude of the spillover effect $A_N\widebar{N}$. 
However, by our construction, $A_{N}\bar{N}$ is set constant independent of $N$ regardless of whether the network is dense ($A_{N} = 1/N$) or sparse ($A_{N}$ is constant), Theorem \ref{theorem2} clarifies that the leading term in the approximation error bound Eq.\ref{eqkl} grows at $\mathcal{O}(N)$. 

Recall from Eq.\ref{equpp} and Eq.\ref{eqkl} that the error due to approximating the welfare at $P$ by the welfare at $Q^*$ can be bounded by $\sqrt{\mathbb{KL}(Q^*\Vert P)}\leq \mathcal{O}(N^{1/2})$. 
If our objective is to maximize $\frac{1}{N}\sum_{i=1}^N \mu_i^P$, Theorem \ref{theorem2} implies that the approximation error of $\frac{1}{N}\sum_{i=1}^N \mu_i^{Q^{\ast}}$ can be bounded from above by $\frac{1}{N}\sqrt{\mathbb{KL}(Q^*\Vert P)}\leq \mathcal{O}(N^{-1/2})$, which converges to zero as $N$ becomes large. 

\begin{remark}\label{remark:32}
We derive the result in Theorem \ref{theorem2} by assuming that the structural parameters $\boldsymbol{\theta}$ are independent of $N$.
Recently, \citet{Joseph2022} discusses the potential issue of using variational approximation when structural parameters grow with the network size $N$, which makes the $C_1$ and $C_2$ grow with $N$. Our regret bound in Theorem \ref{theoremregret} rules out this situation by assuming the Dobrushin's condition.
\end{remark}

\subsection{Implementation}\label{Section:imple}
Having shown how to approximate the average outcome of the Gibbs distribution using the mean field method, this section proposes an algorithm to optimize the approximated welfare in treatment allocation $D$.

Suppose that the set of feasible allocations is subject to a capacity constraint, $\sum_{i=1}^N d_i\leq \kappa$, where $\kappa\in\mathbb{N}_{+}$ specifies the maximum number of units that can be treated. 
We denote the set of feasible allocations by $\mathcal{D}_{\kappa} \equiv \{D\in\{0,1\}^N:\sum_{i=1}^N d_i\leq \kappa\}$, and the approximated welfare by:
\begin{equation}
   \Tilde{W}(D)= \sum_{i=1}^N \Tilde{\mu}_i.
\end{equation}
We seek to maximize the approximated welfare:
\begin{equation}\label{eqmaximization}
\begin{split}
\Tilde{D}\in\argmax_{D\in\mathcal{D}_{\kappa}}\Tilde{W}(D).
\end{split}
\end{equation}
As shown in the Eq.\ref{eqmui}, $\{\Tilde{\mu}_i\}_{i=1}^N$ is a large non-linear simultaneous equation system.
The approximated mean value $\Tilde{\mu}_i$ of each unit $i$ depends upon the approximated mean value $
\Tilde{\mu}_{j}$ and the treatment assignment $d_j$ of her neighbor, unit $j$.
Hence, the optimization problem (Eq.\ref{eqmaximization}) becomes a complicated combinatorial optimization. 

We propose a greedy algorithm (Algorithm \ref{algo}) to solve this problem heuristically.

\medskip

\refstepcounter{algorithm}
\label{algo}
ALGORITHM 2: The adjacency matrix $G$, covariates $\mathcal{X}$, parameter vector $\boldsymbol{\theta}$, and treatment capacity $\kappa$ are given. 
\begin{enumerate}
\item[] \hspace{-0.5cm}Step 1. Initialize the treatment allocation vector at $D^0=(d_1^0,\dots,d_N^0)'=0_{N\times 1},$
so that no unit is initially assigned treatment. For $t=1, \dots, \kappa$, we iterate the following steps: 
\item[] \hspace{-0.5cm}Step 2. Given $D^{t-1} = (d_1^{t-1}, \dots, d_N^{t-1})$, select one untreated unit $i$ in $D^{t-1}$, temporarily set $i$'s treatment status to $d_i\gets 1$,
and denote the resulting candidate treatment vector by $D'$. Compute the associated variationally approximated mean value $\tilde{\mu}$ by running Algorithm \ref{almu} under $D'$ and the marginal welfare gain $\Delta_i := \tilde{W}(D')-\tilde{W}(D^{t-1})$. Repeating for every unit $i$ with $d_i^{t-1}=0$, obtain $\{\Delta_i: d_i^{t-1}=0\}$.
\item[] \hspace{-0.5cm}Step 3. Select the untreated unit that yields the largest welfare improvement: $i^* \gets \arg\max_i \Delta_i.$ 
\item[] \hspace{-0.5cm}Step 4. If $\Delta_{i^*} > 0$, update the treatment allocation $D^{t-1}$ to $D^t = (d_1^t, \dots, d_N^t)$ by assigning treatment to unit $i^*$, $d_{i^*}^{t}=1$ and keep $d_{i}^{t} = d_{i}^{t-1}$ for all $i \neq i^{*}$, and return to Step 2 as long as the capacity is not exhausted ($t<\kappa$). If $\Delta_{i^*} < 0$, terminate the algorithm and output the treatment allocation as $\tilde{D} \gets D^{t-1}$. 
\item[] \hspace{-0.5cm}Step 5. If the iterations of Steps 2 - 4 reach to $t = \kappa$, terminate the algorithm and output the treatment allocation as $\tilde{D} \gets D^{\kappa}$. 
\end{enumerate}

\medskip

Algorithm \ref{algo} is a greedy algorithm that assigns treatment sequentially to the unit contributing the most to the welfare objective, repeating this process until the capacity constraint binds. Specifically, in Step 2, Algorithm \ref{algo} computes the marginal gain from assigning treatment to each untreated unit, and in Step 4, it assigns a treatment to a most influential untreated unit in terms of generating the largest welfare gain. If no unit generates a positive welfare gain, the algorithm stops without exhausing the capacity. We provide a theoretical performance guarantee for this greedy algorithm in Section \ref{sectionawm} and examine its numerical performance in Section \ref{sectionsimul}.

\subsection{Theoretical Analysis}\label{sectionawm}
In this section, we analyze the regret of the treatment allocation rule computed using our greedy algorithm. Let $D^*\in\argmax_{D\in\mathcal{D}_{\kappa}}W(D)$. Then $W(D^*)$ denotes the maximum value of $W(D)$. 
\textit{Regret} is the gap between the maximal equilibrium (oracle) welfare $W(D^*)$ and the equilibrium welfare attained at the treatment allocation rule computed using our greedy algorithm $W(D_G)$. 
We decompose regret into four terms:
 \footnotesize
    \begin{align}
            W(D^*)-W(D_G)&=\underbrace{W(D^*)-\Tilde{W}(D^*)}_{\leq  C_{trans}\sqrt{\mathbb{KL}(Q^*\Vert P)}}+\underbrace{\Tilde{W}(D^*)-\Tilde{W}(\Tilde{D})}_{\leq 0}+\underbrace{\Tilde{W}(\Tilde{D})-\Tilde{W}(D_G)}_{\text{Regret from greedy}}+\underbrace{\Tilde{W}(D_G)-W(D_G)}_{\leq  C_{trans}\sqrt{\mathbb{KL}(Q^*\Vert P)}} \notag \\
            &\leq  2C_{trans}\sqrt{\mathbb{KL}(Q^*\Vert P)}+\Tilde{W}(\Tilde{D})-\Tilde{W}(D_G). \label{eqregret}
        \end{align}
\normalsize
The first term corresponds to the approximation error of using variational approximation;
the second term comes from using the maximizer of the approximated equilibrium welfare $\Tilde{D}$;
the third term comes from using our greedy algorithm instead of using the maximizer of the variationally approximated welfare;
and the last component is again introduced by using the approximated equilibrium welfare $\Tilde{W}(D)$.

By Theorem \ref{theorem2}, the first term in the right-hand side of Eq.\ref{eqregret} can be bounded above at the order of $O(N^{1/2})$. 
The second term in the right-hand side of Eq.\ref{eqregret} captures the regret of the approximated welfare due to greedy optimization. To obtain a nontrivial analytical bound for it, we assume nonnegative treatment and spillover effects as stated in the next assumption. 

\begin{assumption}{(\textbf{Positivity and Monotonicity})}\label{assptse}
    We assume that (i) $\theta_1,\theta_3,\theta_4,\theta_6\geq 0$,
and
        (ii) for any $k\in\mathcal{N}$,
        \begin{equation}\label{asspara1}
        A_N\theta_4\sum_{i\neq k}m_{ik}G_{ik}+\theta_1+X_k'\theta_3\leq 4N.
        \end{equation}
    \end{assumption}
Assumption \ref{assptse} (i) restricts the signs of both own treatment effects and spillover treatment effects, while (ii) constrains the magnitude of the own treatment effects in relation to the network size. 
These sign restrictions on $\bm{\theta}$ are plausible in many applications, such as allocating vaccinations to increase social health, providing discount coupons to stimulate purchases, and assigning tax auditing to reduce tax evasion.

\begin{theorem}{\textbf{Performance Guarantee for greedy Algorithm}:}\label{theoremgreedy}
         Under Assumptions \ref{assumpundir} to \ref{assptse} and the utility function specification of Eq.\ref{eq:alpha} and \ref{eq:beta}, 
         the greedy algorithm enjoys the following approximation guarantee for the problem in Eq.\ref{eqmaximization}:
\begin{equation}\label{Eq:upper}
     \Tilde{W}(\Tilde{D})-\Tilde{W}(D_G)\leq  C_{greedy} \Tilde{W}(\Tilde{D}),
\end{equation}
where $D_G$ is the treatment assignment rule that is obtained by Algorithm \ref{algo}, and $C_{greedy}>0$ is a constant that depends on the curvature and submodularity ratio of the objective function $\Tilde{W}(D)$ (see Appendix for their definitions).
\end{theorem}

This theorem provides a performance guarantee of the greedy optimization, showing that the gap of the approximated welfare value between the global and greedy optima can be bounded by a universal constant fraction of the global maximum. 
This form of the performance guarantee is desirable to have at least, as commonly presented in the literature of integer optimization for submodular functions, e.g., \citet{nemhauser1978analysis} and \citet{bian2017guarantees}, despite that it does not guarantee that the gap asymptotically vanishes.

Combining Eq.\ref{eqkl}, Eq.\ref{eqregret}, and Eq.\ref{Eq:upper}, we obtain the next theorem:
\vspace{-0.2cm}
\begin{theorem}{\textbf{Regret Bound}:}\label{theoremregret}
    Let $D^*$ denote the maximizer of $\Tilde{W}(\mathcal{D})$ and $D_G$ be the assignment vector obtained by Algorithm \ref{algo}. 
    Under Assumptions \ref{assumpundir} to \ref{assptse} and the utility function specification of Eq.\ref{eq:alpha} and \ref{eq:beta}, the regret is bounded from above by:
    
   \begin{equation}
     W(D^*)-W(D_G)\leq  \mathcal{O}(N^{1/2})+C_{greedy}\Tilde{W}(\Tilde{D}).
    \end{equation}
\end{theorem}

Theorem \ref{theoremregret} is our key result. 
It characterizes the guaranteed convergence rate of the overall regret, showing its dependence on the network complexity and the network size. 
The dependence upon the parameters in the utility function is shown implicitly via the terms $C_1$ and $C_2$ in Theorem \ref{theorem2}.  
If we examine the average equilibrium welfare, then the regret bound becomes:

    \begin{equation}
     \frac{1}{N}(W(D^*)-W(D_G))\leq  \mathcal{O}(N^{-1/2})+\frac{1}{N}C_{greedy}\Tilde{W}(\Tilde{D}).
    \end{equation}
The first term is the approximation error and shrinks to zero as $N$ goes to infinity. 
Given that $\Tilde{W}(\Tilde{D})$ can grow proportional to $N$, the per-unit average regret that is associated with our greedy algorithm can converge to a constant.

\section{Data and Applications}\label{sec:danda}
We illustrate our proposed method using the dataset of \citet{banerjee2013diffusion}, which examines take-up of a microfinance initiative in India.\footnote{The dataset is available at \href{https://doi.org/10.7910/DVN/U3BIHX}{https://doi.org/10.7910/DVN/U3BIHX}.} 
A detailed description of the study can be found in the original paper. 
This study features 43 villages in Karnataka that participated in a newly available microfinance loan program. 
Bharatha Swamukti Samsthe (BSS)---an Indian non-governmental microfinance institution administering the initiative---provided information about the availability of microfinance and program details (the treatment) to individuals that they identified as `leaders' (e.g., teachers, shopkeepers, savings group leaders, etc.) so as to maximize the number of households that chose to adopt the microfinance product. 
The data provide network information at the household level (network data is available across 12 dimensions, including financial and medical links, social activity, and known family members) for each village. 
We use all available households characteristics that are available in the dataset (quality of access to electricity, quality
of latrines, number of beds, number of rooms, the number of beds per capita, and the number
of rooms per capita) as covariates.
The program started in 2007, and the survey for microfinance adoption was finished in early 2011. 
We treat each household's choice about whether to purchase microfinance or not as observations drawn from a stationary distribution of the sequential game.

The most common occupations in these villages are in agriculture, sericulture, and dairy production \citep{banerjee2013diffusion}. 
In addition, these villages had almost no exposure to microfinance institutions and other types of credit before this program. 
We allow the parameters of our utility function to be different across villages, and estimate them for each village using Markov Chain Monte Carlo Maximum Likelihood \citep{snijders2002markov}.\footnote{MCMC Maximum Likelihood evaluates the score function by MCMC at each Newton-Raphson step. The number of likelihood evaluations, determined by a convergence criterion for the parameter sequence in Newton-Raphton steps (i.e., $ \| \theta^{t-1} - \theta^t \|$ is sufficiently small), is typically much smaller than $N\kappa$, the number of welfare evaluations needed for greedy optimization for treatment allocation.}\label{footmcmc} We also note that each household is connected to approximately $10$ others on average across all of the $43$ villages.
Comparing this with the total number of households in each village (there are between $107$ and $341$ households in each village), we find that the household network in each village is a sparse network. 
We, therefore, choose $A_N=1$ in our simulation and application. In addition, we choose $m(X_i,X_j) = \frac{1}{1+\vert X_i-X_j\vert}$, which is a monotonically decreasing function in the metric between $X_i$ and $X_j$.

\subsection{Simulation Exercises}\label{sectionsimul}
In this section, we evaluate the performance of our greedy algorithm in simulation exercises. We consider small networks of size $N=5,7,9,11,13$ or $15$ to compute the distribution of outcomes at every possible assignment vector, and use a brute force method to find an optimal treatment allocation.

We design the simulation setting to be relevant to empirical settings. For each choice of $N$, we randomly select $N$ households from Village~1 (as indexed in the dataset of \citet{banerjee2013diffusion}) and treat them as the population of a small network, maintaining the covariates of the sampled households and setting the estimated structural parameters in Village 1 as the truth. In estimation of the structural parameters, we adopt the specification of the utility function (Eqs.~\ref{eq:alpha} and \ref{eq:beta}). The parameter estimates in Village 1 satisfies Dobrushin's condition.
To match the network structure with the empirical application, we estimate a binary regression of network formations, $G_{ij}=\mathbbm{1}\{\beta_0+\beta_1\vert X_i-X_j \vert+\epsilon_{ij}\geq 0\},$ $\epsilon_{ij} \sim_{iid} Logistic$, using the whole observations from Village~1. 
We then generate $100$ networks of size $N$ (keeping the sampled households fixed) from the estimated network formation model. For each of these networks, we compute an optimal treatment allocation with our method subject to the capacity constraint of $\kappa = \lceil0.3N\rceil$, and simulate the corresponding equilibrium welfare. We then report its average over these $100$ networks to assess the performance of our method. 
We then repeat the same simulation exercise for Village~4 (as indexed in the dataset of \citet{banerjee2013diffusion}). The reason that we chose Village 4 is because the structural parameter estimates of Village 4 does \textit{not} satisfy Dobrushin’s condition. So this case in comparison to Village 1 offers an illustration of how our proposed method performs when Dobrushin's condition fails.

First, we consider all possible treatment allocations subject to the capacity constraint and perform brute force search to find an optimal assignment. 
For instance, when $N=15$, the number of feasible assignment vectors is $32,768$. 
We compute the joint distribution of outcomes at each possible treatment allocation by applying the joint probability mass function of the Gibbs distribution (Eq.\ref{eqjointdis}). 
Second, to assess the welfare loss from implementing the variational approximation, we evaluate the regret of a treatment assignment rule that is obtained by maximizing the variationally approximated welfare over every feasible treatment allocation meeting the capacity constraint (without greedy optimization). 
We label this method of obtaining the optimal treatment assignment as \textit{brute force with variational approximation} (BFVA). 
Third, we perform random allocation that assigns treatment to a fraction $\kappa$ of units independently of personal characteristics and network structure. Fourth, we compute the Bonacich centralities\footnote{Following \cite{ballester2006s}, Bonacich centrality is defined as follows:
      consider a network \(g\) with adjacency \(n\)-square matrix \(G\) and a scalar \(a\) such that $M(g,a) = \bigl[\mathbf{I} - aG\bigr]^{-1}$
is well‐defined and nonnegative, where $\mathbf{I}$ denotes the $N$-square identity matrix. The vector of Bonacich centralities of parameter \(a\) in \(g\) is
$
b(g,a) = \bigl[\mathbf{I} - aG\bigr]^{-1}\,\mathbf{1}.
$
Given that $a$ is a nuisance parameter, we set $a=0.1$, following \citet{galeotti2020targeting}.} for each network, and assign treatment to units in the order of decreasing centrality until the capacity constraint binds. We labeled this method of obtaining the optimal treatment assignment as \textit{Centrality}.


Figure \ref{figure1} shows the simulated average performances of the three aforementioned methods and the greedy targeting rule in terms of in-sample average welfare, with details summarized in Table \ref{table brute} in the Appendix. On the left-hand side of Figure \ref{figure1}, we compute the regret for each method scaled by the average welfare level of the global optimal allocation. 
The right-hand side of Figure \ref{figure1} shows the results obtained using the parameters estimated from Village 4, which allow us to evaluate performance when Assumption \ref{ass:boundspill} is violated.

Figure \ref{figure1} shows that our greedy algorithm performs as well as the brute-force method in a small-network setting, indicating strong performance of our approach.
Moreover, our greedy algorithm can achieve the same performance as BFVA, which means that using our greedy algorithm has a negligible effect upon regret. We also observe that targeting based on Bonacich centrality performs similarly to random allocation.

From the right-hand side of Figure \ref{figure1}, we observe that our method continues to perform as well as the optimal allocation. The regret associated with centrality and random allocation methods increases when Assumption \ref{ass:boundspill} is violated.

\begin{figure}[h]
\centering
\begin{minipage}[t]{0.48\textwidth}
\centering
\includegraphics[width=7cm,height=5cm]{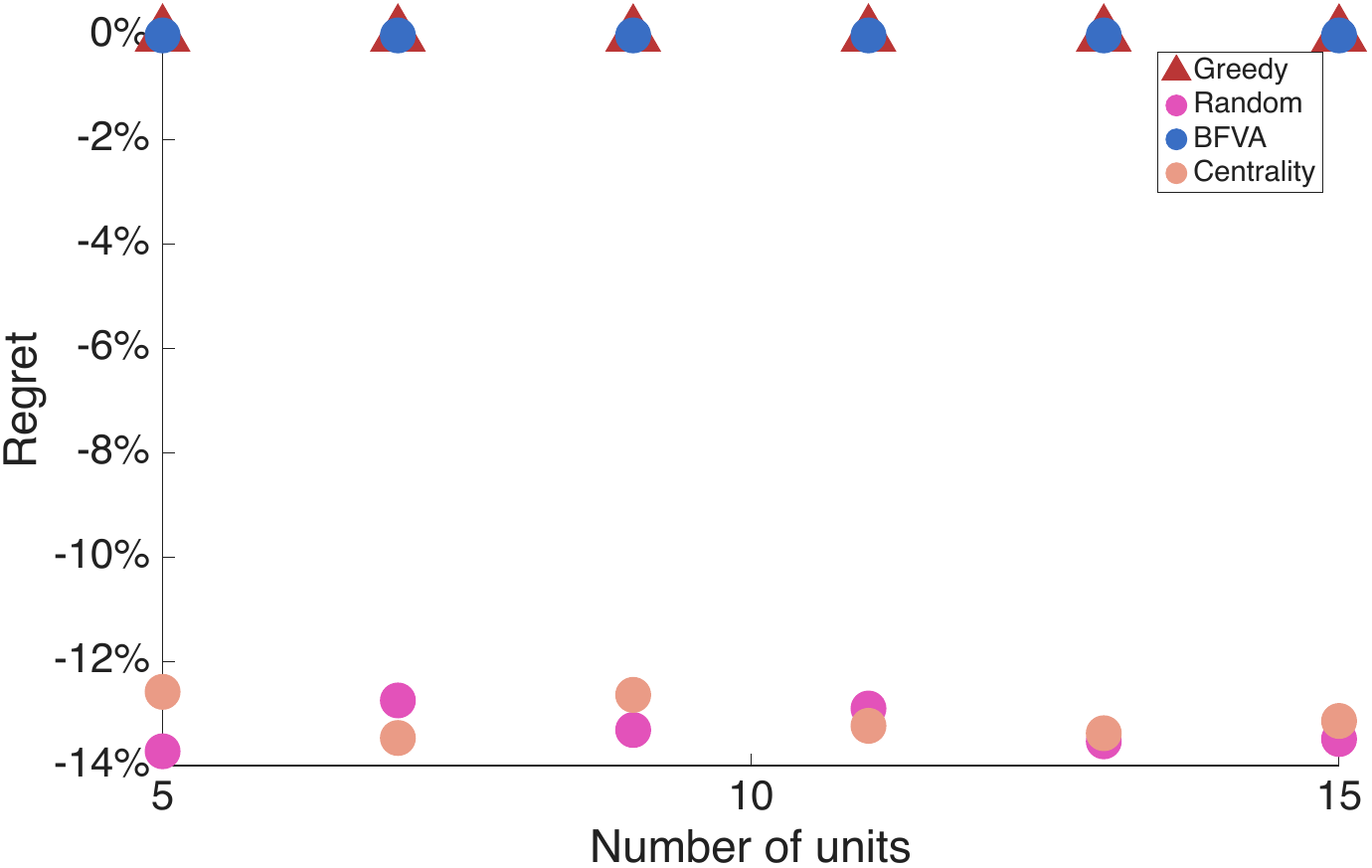}
\end{minipage}
\begin{minipage}[t]{0.48\textwidth}
\centering
\includegraphics[width=7cm,height=5cm]{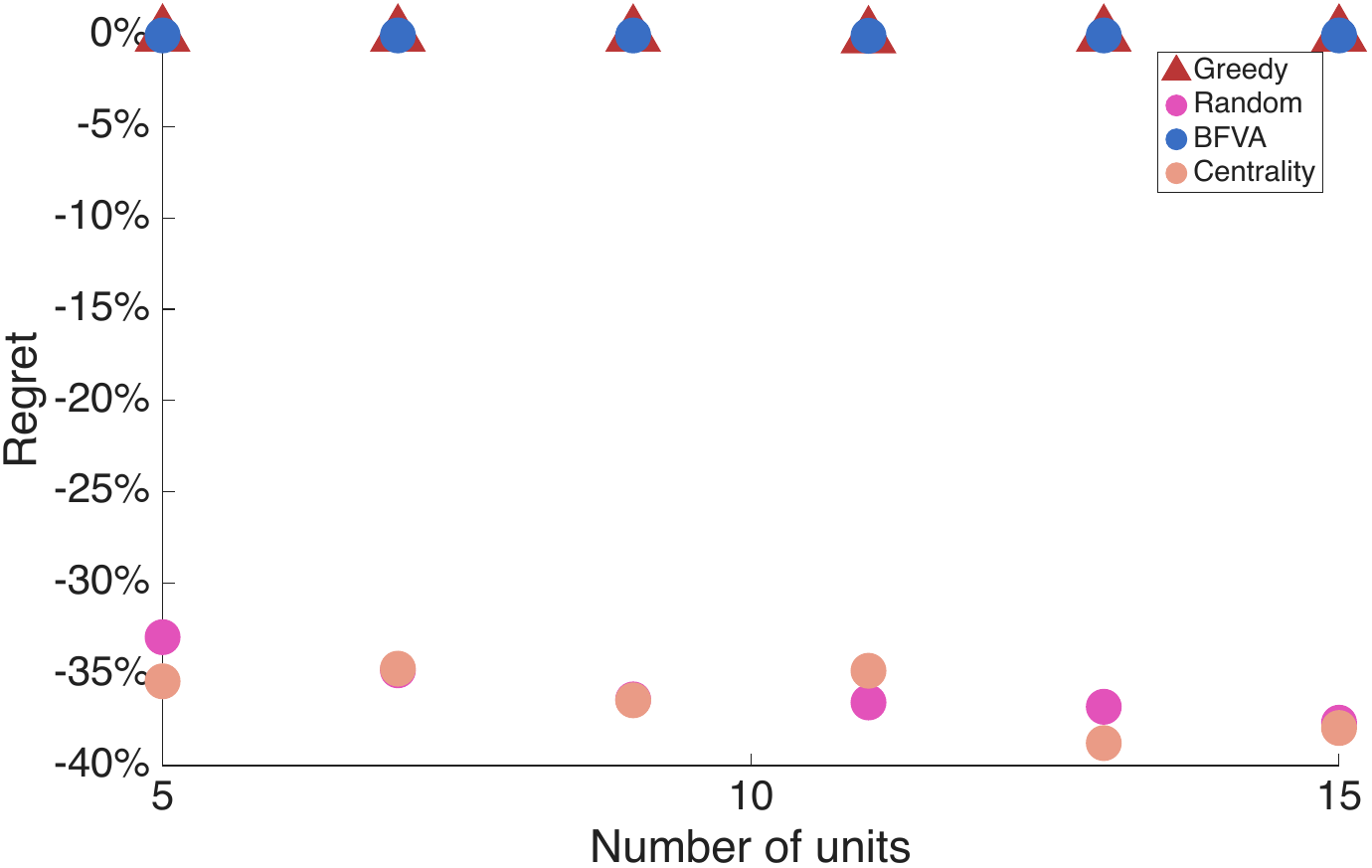}
\end{minipage}
\caption{Average simulated performance of four treatment allocation rules: greedy optimization with variational approximation (Greedy), random assignment (Random), brute force search with variational approximation (BFVA), and an assignment based on Bonacich centrality (Centrality). The left panel reports regret relative to the global optimal allocation using parameters estimated from village 1. The right panel reports the same comparison using parameters estimated from Village 4, where Assumption \ref{ass:boundspill} is violated.}
\label{figure1}
\end{figure}

\subsection{Empirical Application}\label{sectionappli}

Our target in this application is to maximize the participation rate of microfinance (4 years after program assignment) given a capacity constraint on treatment (i.e., Eq.\ref{eqobj});
we set our capacity constraint equal to the number of households that BSS contacted in the original study. 
For each iteration of the procedure, we set the number of draws in the Gibbs sampling procedure equal to $200N^2$.  

Table \ref{tableemprical} shows the average probability of taking up microfinance observed in each village (column Sample Avg.) and the prediction of village-level take-up probability obtained from our MCMC-MLE estimates (column Welfare under Original).
We refer to the former as the \textit{Sample Average} and the latter as the \textit{Welfare under Original Allocation}.
We provide standard errors for the Sample Average, which we calculate using network HAC estimation (\citet{leung2019inference}; and \citet{kojevnikov2021limit}).
We compute Welfare under Original Allocation by substituting the estimated parameters and the original treatment allocation (used by BSS) into our model.
To further evaluate the performance of our proposed method, we randomly draw 100 treatment allocations that satisfy the capacity constraint in each village, and calculate the probability of purchasing microfinance for each allocation.
We refer to the average probability over these draws as the \textit{Welfare under Random Allocation}. In addition, we compute welfare based on treatment assignments determined by Bonacich centrality for each network, which we denote as the \textit{Welfare under Centrality}. 
We then implement our proposed method with the estimated parameters to find the optimal treatment allocation rule. 
We refer to the share of households adopting microfinance according to the optimal treatment allocation and our model as the \textit{Welfare under 
 Greedy Allocation}.
Table \ref{tableemprical} records these statistics for the $43$ villages in the dataset, with the final column comparing the Welfare under our Greedy Allocation with the Welfare under Original Allocation. It also contains bootstrap standard errors for welfare based on 100 bootstrap samples for each village. The bootstrap samples are obtained by drawing outcomes from the MCMC stationary distribution simulated under the structural parameters estimated with the original sample. To obtain the standard errors for the Welfare under Greedy Allocation, for each bootstrap sample, we estimate the structural parameters and simulate welfare under a greedy optimal allocation.\footnote{If an optimal welfare is differentiable in the structural parameters and the greedy optimal allocation under the true structural parameter value is unique, we expect the bootstrap delta method applies to yield asymptotically valid standard errors.}

First, we note that the estimated average share of households who adopt microfinance under the MCMC-MLE estimates fits the data well for all $43$ villages. 
Second, we find that the centrality-based allocation delivers a level of welfare comparable to both the original treatment allocation and the random allocation. This result suggests that the centrality-based approach may fail to capture the spillover effects arising from strategic interactions. It also implies that the leaders that BSS selected were not particularly effective in encouraging take-up by other households.
Third, we find that our proposed method compares favourably to the method that is implemented in \citet{banerjee2013diffusion}, yielding a treatment allocation that attains a higher welfare level.
As shown in Table \ref{tableemprical}, the welfare gain is positive for all $43$ villages (exceeding $100\%$ in some villages).
This indicates that if the specification of the sequential network game is correct in the context of the current application, individualized treatment allocation that takes into account network spillovers can generate large welfare gains.
Existing empirical work around social networks has not quantified the welfare gain from individualized treatment allocation under spillovers due to a lack of feasible procedures to obtain an optimal individualized assignment policy. 
In contrast, we uncover evidence of the significant welfare gains that can be realized by exploiting network spillovers. 

To highlight the computational advantage of our proposed method, we compare the computation time for equilibrium welfare by the variational approximation and the one by MCMC. For a village with $175$ individuals under a fixed set of estimated parameters, computing an approximate welfare value for a given treatment allocation takes $0.04$ seconds using the variational approximation, whereas it takes $15$ seconds to approximate the welfare using MCMC with 50,000 draws. 
Since the welfare criterion must be evaluated $N \cdot \kappa \approx 4{,}500$ times when $\kappa = \lceil 0.3N \rceil$, the time difference of evaluating the welfare once between the variational approximation and MCMC amounts to approximately $20$ hours for obtaining an optimal treatment allocation.

\citet{akbarpour2025just} shows that a targeted treatment allocation under a capacity constraint is dominated by a random allocation under a slightly relaxed capacity constraint.
\citet{beaman2021can} points out that the result of \citet{akbarpour2025just} holds only under the following three conditions: (i) agents must adopt a new behavior when exposed to other agent who has adopted the behavior, (ii) the time period for social interactions is sufficiently long, and (iii) the interactions within the network are frequent.
Our results do not contradict the result of \citet{akbarpour2025just} since condition (i) does not hold in our framework; agents make decisions of adopting or not adopting strategically subject to idiosyncratic shocks.

\begin{table}[h!]
\setstretch{1}
\begingroup
\setlength{\tabcolsep}{3pt} 
\renewcommand{\arraystretch}{.3}
 \begin{adjustwidth}{0cm}{}
 \linespread{1}
\footnotesize
\centering 
\begin{threeparttable}
 \begin{tabular}{@{}cc|cccc|cc@{}}
   \toprule
    \multirow{2}{*}{\textit{Village}}
    & \multicolumn{1}{c}{\multirow{2}{*}{\textit{Sample Avg. \scriptsize(s.e.)}}}
    & \multicolumn{4}{c}{\textit{Welfare under}}
    & \multicolumn{2}{c}{\textit{Welfare Gain$^*$}}\\
    &&\multicolumn{1}{c}{\textit{Original \scriptsize(s.e.)}} & \textit{Random \scriptsize(s.e.)} 
    & \textit{Centrality \scriptsize(s.e.)} 
    & \textit{Greedy \scriptsize(s.e.)} 
    & \textit{Level \scriptsize(s.e.)} 
    & \textit{$\%$}\\
  \midrule
 \textbf{1}& $.24$ ($.03$)&$.24$ ($.03$)&$.24$ ($.03$)&$.24$ $(.03)$&$.27$ ($.04$)&$.03$ ($.02$)& $12.97\%$\\
 \textbf{2}&$.15$ ($.02$)&$.17$ ($.03$)&$.19$ ($.03$)&$.18$ $(.03)$&$.22$ ($.04$)&$.05$ ($.02$)& $29.73\%$\\
 \textbf{3}&$.14$ ($.02$)&$.13$ ($.02$)&$.13$ ($.03$)&$.13$ $(.02)$&$.17$ ($.05$)&$.03$ ($.03$)& $25.57\%$\\
 \textbf{4}&$.08$ ($.02$)&$.12$ ($.03$)&$.12$ ($.02$)&$.19$ $(.03)$&$.21$ ($.03$)&$.10$ ($.02$)& $83.88\%$\\
 \textbf{5}&$.23$ ($.03$)&$.21$ ($.04$)&$.21$ ($.03$)&$.20$ $(.04)$&$.26$ ($.04$)&$.05$ ($.02$)& $23.70\%$\\
 \cmidrule(lr){2-2}\cmidrule(lr){3-6}\cmidrule(lr){7-8}
 \textbf{6}&$.18$ ($.03$)&$.29$ ($.02$)&$.28$ ($.03$)&$.31$ $(.02)$&$.36$ ($.05$)&$.07$ ($.03$)& $23.37\%$\\
 \textbf{7}&$.30$ ($.04$)&$.26$ ($.04$)&$.28$ ($.03$)&$.27$ $(.04)$&$.35$ ($.04$)&$.09$ ($.02$)& $34.03\%$\\
 \textbf{8}&$.12$ ($.03$)&$.13$ ($.03$)&$.14$ ($.03$)&$.14$ $(.03)$&$.19$ ($.04$)&$.06$ ($.02$)& $45.97\%$\\
 \textbf{9}&$.21$ ($.03$)&$.19$ ($.02$)&$.19$ ($.02$)&$.19$ $(.03)$&$.21$ ($.03$)&$.02$ ($.02$)& $10.52\%$\\
 \textbf{10}&$.36$ ($.04$)&$.39$ ($.04$)&$.43$ ($.04$)&$.35$ $(.05)$&$.50$ ($.05$)&$.10$ ($.03$)& $26.55\%$\\
\cmidrule(lr){2-2}\cmidrule(lr){3-6}\cmidrule(lr){7-8}
 \textbf{11}&$.45$ ($.05$)&$.50$ ($.04$)&$.50$ ($.04$)&$.50$ $(.04)$&$.54$ ($.04$)&$.03$ ($.02$)& $6.85\%$\\
 \textbf{12}&$.15$ ($.02$)&$.16$ ($.02$)&$.15$ ($.02$)&$.18$ $(.02)$&$.27$ ($.05$)&$.11$ ($.03$)& $69.88\%$\\
 \textbf{13}&$.19$ ($.02$)&$.22$ ($.03$)&$.22$ ($.03$)&$.23$ $(.03)$&$.27$ ($.04$)&$.06$ ($.02$)& $25.15\%$\\
 \textbf{14}&$.17$ ($.03$)&$.18$ ($.03$)&$.18$ ($.02$)&$.18$ $(.03)$&$.21$ ($.03$)&$.03$ ($.01$)& $18.12\%$\\
 \textbf{15}&$.27$ ($.02$)&$.28$ ($.02$)&$.29$ ($.02$)&$.29$ $(.03)$&$.34$ ($.03$)&$.06$ ($.02$)& $22.46\%$\\
 \cmidrule(lr){2-2}\cmidrule(lr){3-6}\cmidrule(lr){7-8}
 \textbf{16}&$.35$ ($.04$)&$.47$ ($.05$)&$.46$ ($.04$)&$.48$ $(.04)$&$.54$ ($.04$)&$.07$ ($.03$)& $15.75\%$\\
 \textbf{17}&$.19$ ($.03$)&$.21$ ($.02$)&$.20$ ($.03$)&$.20$ $(.02)$&$.25$ ($.03$)&$.04$ ($.01$)& $16.81\%$\\
 \textbf{18}&$.19$ ($.03$)&$.19$ ($.04$)&$.20$ ($.03$)&$.19$ ($.04$)&$.21$ ($.04$)&$.01$ ($.03$)& $6.14\%$\\
 \textbf{19}&$.08$ ($.02$)&$.10$ ($.02$)&$.11$ ($.02$)&$.10$ ($.02$)&$.13$ ($.03$)&$.03$ ($.02$)& $33.41\%$\\
 \textbf{20}&$.19$ ($.02$)&$.21$ ($.02$)&$.20$ ($.02$)&$.25$ ($.04$)&$.34$ ($.05$)&$.13$ ($.04$)& $63.50\%$\\
 \cmidrule(lr){2-2}\cmidrule(lr){3-6}\cmidrule(lr){7-8}
 \textbf{21}&$.35$ ($.04$)&$.34$ ($.04$)&$.34$ ($.03$)&$.32$ ($.04$)&$.42$ ($.04$)&$.08$ ($.02$)& $23.19\%$\\
 \textbf{22}&$.25$ ($.04$)&$.26$ ($.04$)&$.25$ ($.02$)&$.26$ ($.04$)&$.30$ ($.02$)&$.04$ ($.01$)& $15.33\%$\\
 \textbf{23}&$.21$ ($.03$)&$.22$ ($.03$)&$.22$ ($.03$)&$.21$ ($.03$)&$.29$ ($.03$)&$.07$ ($.02$)& $29.64\%$\\
 \textbf{24}&$.24$ ($.03$)&$.21$ ($.03$)&$.19$ ($.03$)&$.24$ ($.03$)&$.27$ ($.03$)&$.07$ ($.01$)& $31.81\%$\\
 \textbf{25}&$.23$ ($.02$)&$.26$ ($.04$)&$.26$ ($.04$)&$.26$ ($.04$)&$.27$ ($.06$)&$.01$ ($.03$)& $4.08\%$\\
 \cmidrule(lr){2-2}\cmidrule(lr){3-6}\cmidrule(lr){7-8}
 \textbf{26}&$.19$ ($.04$)&$.21$ ($.03$)&$.22$ ($.03$)&$.21$ ($.03$)&$.27$ ($.03$)&$.06$ ($.01$)& $29.24\%$\\
 \textbf{27}&$.09$ ($.02$)&$.12$ ($.02$)&$.13$ ($.02$)&$.12$ ($.01$)&$.16$ ($.03$)&$.04$ ($.01$)& $29.91\%$\\
 \textbf{28}&$.12$ ($.03$)&$.10$ ($.02$)&$.13$ ($.02$)&$.09$ ($.02$)&$.34$ ($.04$)&$.24$ ($.03$)& $233.82\%$\\
 \textbf{29}&$.10$ ($.02$)&$.08$ ($.01$)&$.08$ ($.01$)&$.08$ ($.02$)&$.12$ ($.02$)&$.04$ ($.02$)& $54.14\%$\\
 \textbf{30}&$.11$ ($.05$)&$.17$ ($.04$)&$.15$ ($.03$)&$.16$ ($.04$)&$.24$ ($.06$)&$.07$ ($.03$)& $41.43\%$\\
 \cmidrule(lr){2-2}\cmidrule(lr){3-6}\cmidrule(lr){7-8}
 \textbf{31}&$.15$ ($.02$)&$.21$ ($.04$)&$.20$ ($.04$)&$.22$ ($.04$)&$.29$ ($.05$)&$.08$ ($.02$)& $40.00\%$\\
 \textbf{32}&$.08$ ($.02$)&$.12$ ($.02$)&$.12$ ($.02$)&$.16$ ($.03$)&$.23$ ($.03$)&$.12$ ($.03$)& $103.30\%$\\
 \textbf{33}&$.15$ ($.02$)&$.12$ ($.02$)&$.12$ ($.02$)&$.11$ ($.02$)&$.13$ ($.03$)&$.01$ ($.03$)& $7.25\%$\\
 \textbf{34}&$.18$ ($.04$)&$.25$ ($.03$)&$.25$ ($.04$)&$.25$ ($.03$)&$.28$ ($.05$)&$.03$ ($.02$)& $11.77\%$\\
 \textbf{35}&$.11$ ($.02$)&$.12$ ($.02$)&$.12$ ($.02$)&$.12$ ($.02$)&$.16$ ($.02$)&$.04$ ($.01$)& $34.18\%$\\
 \cmidrule(lr){2-2}\cmidrule(lr){3-6}\cmidrule(lr){7-8}
 \textbf{36}&$.17$ ($.02$)&$.19$ ($.02$)&$.20$ ($.02$)&$.20$ ($.02$)&$.24$ ($.02$)&$.04$ ($.01$)& $20.42\%$\\
 \textbf{37}&$.30$ ($.04$)&$.34$ ($.03$)&$.34$ ($.04$)&$.35$ ($.02$)&$.40$ ($.05$)&$.06$ ($.02$)& $18.29\%$\\
 \textbf{38}&$.15$ ($.03$)&$.21$ ($.04$)&$.20$ ($.04$)&$.19$ ($.05$)&$.25$ ($.05$)&$.04$ ($.02$)& $18.12\%$\\
 \textbf{39}&$.21$ ($.03$)&$.18$ ($.03$)&$.16$ ($.02$)&$.17$ ($.03$)&$.22$ ($.02$)&$.04$ ($.01$)& $23.12\%$\\
 \textbf{40}&$.16$ ($.03$)&$.18$ ($.02$)&$.19$ ($.02$)&$.26$ ($.03$)&$.31$ ($.03$)&$.13$ ($.02$)& $70.85\%$\\
 \cmidrule(lr){2-2}\cmidrule(lr){3-6}\cmidrule(lr){7-8}
 \textbf{41}&$.16$ ($.04$)&$.18$ ($.03$)&$.19$ ($.03$)&$.19$ ($.03$)&$.28$ ($.05$)&$.10$ ($.04$)& $55.28\%$\\
 \textbf{42}&$.18$ ($.03$)&$.17$ ($.03$)&$.16$ ($.03$)&$.17$ ($.03$)&$.24$ ($.04$)&$.07$ ($.02$)& $40.34\%$\\
 \textbf{43}&$.24$ ($.04$)&$.27$ ($.03$)&$.30$ ($.04$)&$.28$ ($.03$)&$.42$ ($.06$)&$.15$ ($.04$)& $53.93\%$\\
\bottomrule
   \end{tabular}
   \end{threeparttable}
   \caption{Welfare performance comparison using $43$ Indian villages microfinance data from \cite{banerjee2013diffusion};
   $\ast$ Welfare Gain compares the equilibrium welfare simulated under Greedy Allocation and the equilibrium welfare simulated under Original Allocation implemented by BSS.}
\label{tableemprical}
   \end{adjustwidth}
   \endgroup
\end{table}

\section{Conclusion}\label{sectionconclusion}
In this work, we have introduced a novel method to obtain individualized treatment allocation rules that maximize the equilibrium welfare in sequential network games. 
To handle the analytical and computational challenges of analyzing the stationary distribution of outcomes, we use variational approximation and maximize the approximated welfare criterion using a greedy maximization algorithm over treatment allocations. 
We bound the welfare regret, taking into account the approximation errors of the variational approximation and of the greedy maximization. 

There are several questions remained to be studied.
First, we assume that the network structure does not change in response to the treatment allocation. One could apply the framework of \citet{badev2021nash} to incorporate joint determination of network formation and individual choices into policy targeting, but we leave the details for future research.
Second, we may want to perform inference for the welfare at the obtained assignment rule, taking into account the uncertainty of parameter estimates and a potential winner's bias \citep{andrews2020inference}. 
Third, we have used a naive mean field method in this work. 
As is mentioned in \citet{wainwright2008graphical}, using a structural mean field method can improve the performance of an approximation and can lead to better welfare performance. Fourth, we assume that the network structure in our data is perfectly measured, which may not hold in practice. This issue has been examined in \citet{lewbel2025estimating}.

\appendix

\bigskip
\section{Main results}
In this section of the Appendix, we provide the proofs for the main results of Section \ref{sectionall}. We first denote the matrix norms induced by vector norms as $\Vert A\Vert_{a,b}\coloneqq \sup\{\Vert Ax\Vert_{b}:\Vert x\Vert_a\leq 1\}$. 
Let $p$ be a generic measure with support $\mathcal{Y}^N$, and denote the conditional distribution of $Y_i$ given $Y_{-i}$ as $p_i(Y_i\vert Y_{-i})\coloneqq \frac{p(Y)}{\sum_{Y_{i}\in\mathcal{Y}}p(Y_i,Y_{-i})}$ given the choices of the other units $Y_{-i}$. A matrix $A=(a_{ij})_{i,j\le N}$ is a coupling matrix if it satisfies $a_{ii}=0$ for all $i$ and for $i\neq j$
  \[
    \bigl\|p_i(\cdot \mid Y_{-i}) 
         \;-\;
         p_i(\cdot \mid Y_{-i}')\bigr\|_{\mathrm{TV}}
    \;\le\;
    a_{ij},
  \]
  whenever $Y,Y'\in \mathcal{Y}^N$ differ only at the $j$-th coordinate.
\subsection{Proof of Lemma \ref{lemma:eb}}\label{lemmaproof:eb}
\begin{proof}
    For any coupling $\omega \in \Omega(P,Q)$ (i.e.\ a joint distribution on $(Y,Y') \in \mathcal{Y}^N \times \mathcal{Y}^N$
with marginals $P$ and $Q$), define the Hamming cost
\[
\mathbb{E}_\omega\Bigl[\sum_{i=1}^N \mathbf{1}\{Y_i \neq Y_i'\}\Bigr]
~=~
\sum_{i=1}^N \mathbb{E}_\omega[\mathbf{1}\{Y_i \neq Y_i'\}]
~=~
\sum_{i=1}^N Pr_\omega(Y_i \neq Y_i').
\]
For each coordinate $i$, given $Y_i$ is a Bernoulli random variable, we have
\[
Pr_\omega(Y_i \neq Y_i')
~\ge~
\bigl| P (Y_i=1)-Q(Y_i=1)\bigr|.
\]
Summing over $i=1,\ldots,N$, we get
\[
\mathbb{E}_\omega\Bigl[\sum_{i=1}^N \mathbf{1}\{Y_i \neq Y_i'\}\Bigr]
~
\ge~
\sum_{i=1}^N 
\bigl|
P(Y_i=1)
-
Q(Y_i=1)
\bigr|.
\]
Since $W_1(P,Q)$ is the infimum of the above expected Hamming cost over all couplings
$\gamma \in \Omega(P,Q)$, the result follows:
\[
W_1(P,Q)
=
\inf_{\omega \in \Omega(P,Q)}
\mathbb{E}_\omega\Bigl[\sum_{i=1}^N \mathbf{1}\{Y_i \neq Y_i'\}\Bigr]
~
\ge~
\sum_{i=1}^N \bigl|P(Y_i=1) - Q(Y_i=1)\bigr|.
\]
\end{proof}

\subsection{Proof of Lemma \ref{pro:tala}}\label{lemmaproof:tI}
\begin{proof}
We aim to show: 
\begin{equation} 
        W_1(P,Q)\leq C_{trans}\sqrt{\mathbb{KL}(Q\Vert P)},
    \end{equation}
  The proof consists of the following two steps. First, we show that our stationary distribution \(P\) satisfies the weak transport inequality $\overline{\mathrm T}_f(Q\mid P)\le \mathbb{KL}(Q\mid P)$,
    where \(\overline{\mathrm{T}}_f(Q\mid P)\) denotes the weak transport cost from Definition~\ref{def:5.1}. Second, we show that $W_1(P,Q)\le \sqrt{\frac{1}{c_{\mathrm{trans}}}\,\overline{\mathrm{T}}_f(Q\mid P)}$.
    
    \begin{definition}{\textbf{Weak transport cost}}\label{def:5.1}:
  Let $p$ and $q$ be probability measures on $\mathcal{Y}^N$, define 
  the weak transport cost between $p$ and $q$ as 
 \begin{equation}
    \overline{\mathrm{T}}_f(q\mid p)=\inf_{\pi\in\Omega(p,q)}
      \int_{\mathcal{Y}^N}
        f(
          y - \int_{\mathcal{Y}^N} y'\, \pi(dy'\vert y)
        )
      \,p(\mathrm{d}y).
\end{equation}
where $f:\mathcal{Y}^N\rightarrow[0,+\infty]$ is a lower-semicontinuous convex function, and the infimum is taken over all couplings $\pi$ of $p$ and $q$ 
  (i.e.\ measures on $\mathcal{Y}^N \times \mathcal{Y}^N$ with marginals 
  $p$ and $q$). For each $y \in \mathcal{Y}^N$, $\pi(dy'\vert y)$ denotes 
  the conditional measure satisfying $\pi(\mathrm{d}y'\,\mathrm{d}y) 
  = \pi(dy'\vert y)\,p(\mathrm{d}y)$ ($p$‐almost surely). We will say that $p$ satisfies the \textbf{weak transport inequality} if for every probability
measure $q$ on $\mathcal{Y}^N$,
    \begin{equation}\label{eq:tfineq}
        \max\big(\overline{\mathrm{T}}_f(q \mid p),\overline{\mathrm{T}}_f(p \mid q)\big)\leq \mathbb{KL}(q\mid p).
    \end{equation}
\end{definition}
    To prove the weak transport inequality, we use Lemma~\ref{thm:5.3} from \citet{gozlan2017kantorovich}, which establishes the equivalence between the weak transport inequality and the dimension-free convex concentration property.\footnote{For comparison, \citet{gozlan2009characterization} establish the equivalence between the dimension-free concentration property and Talagrand’s transportation inequality (Eq.~\ref{eq:talagrand}). In particular, when \(Y\) is an i.i.d.\ Gaussian vector, the dimension-free concentration inequality of \citet{sudakov1978extremal} holds for all Lipschitz functions. It follows that Talagrand’s transportation inequality holds in that setting, which in turn implies the weak transport inequality since
\[
  \overline{\mathrm{T}}_f(Q\mid P)\le W_2^2(P,Q).
\]
By contrast, for dependent random vectors, the dimension-free concentration property generally fails \citep{gotze2019higher}, so Talagrand’s transportation inequality is not available in our setting. Instead, such distributions satisfy the dimension-free \emph{convex} concentration inequality by Lemma~\ref{lemma:a54}.} It therefore remains to show that our stationary distribution $P$ satisfies the dimension-free convex concentration property.
    \begin{lemma}{\citep[\S Corollary~5.11, Theorem~8.8, 8.15]{gozlan2017kantorovich}}\label{thm:5.3}
  Let $Y$ be a random vector in $\mathcal{Y}^N$ with distribution~$p$. 
  The following conditions are equivalent:
  \begin{enumerate}
    \item There exists $K$ such that $Y$ has the dimension-free 
      convex concentration property with constant~$K$.
    \item There exists $c$ such that $p$ satisfies the weak transport inequality (Eq.\ref{eq:tfineq}) with $f(y) = c\Vert y\Vert_1^2$.
  \end{enumerate}
\end{lemma}
\begin{definition}{\textbf{Convex Concentration Property:}}
    A random vector $Y$ in $\mathbb{R}^N$ has a convex concentration 
property if there exists a constant $0<K<\infty$ such that for any $L$-Lipschitz convex function 
$g: \mathbb{R}^N \to \mathbb{R}$ and any $t > 0$, it holds
    \begin{equation}\label{eq:convexcon}
        \Pr\bigl(\lvert g(Y) - \mathrm{Med}\,g(Y)\rvert \ge t\bigr)\leq 2 \exp\big(-t^2/K^2L^2\big). 
    \end{equation}\end{definition}
To show the dimension-free convex concentration property of our stationary distribution $P$, we apply the result from Lemma \ref{lemma:a54}, which guarantees the stationary distribution $P$ satisfies the dimension-free convex concentration inequality (Eq.\ref{eq:convexcon}) if it satisfies the approximate tensorization (Definition \ref{def:AT}).
\begin{lemma}{(\citet[\S Proposition~5.4]{adamczak2019note})}\label{lemma:a54}
If $Y$ is a $[-1,1]^N$ - valued random vector with law $p$, which satisfies the approximate tensorization $\mathrm{AT}(C)$,
  then $Y$ satisfies the dimension-free convex concentration inequality (Eq.\ref{eq:convexcon})
  with constant $K$ depending only on $C$.
\end{lemma}
\begin{definition}{\textbf{Approximate tensorization of entropy:}}\label{def:AT}
  We say that a measure $p$ on $\mathcal{Y}^N$ has the \emph{approximate tensorization property} 
  with constant $C$ (abbreviated as $\mathrm{AT}(C)$) if for every function 
  $g:\mathcal{Y}^N\to[0,\infty)$,
  \[
    \mathrm{Ent}_p(g)
    \;\le\;
    C\,\mathbb{E}_p\Bigl[\;\sum_{i=1}^N 
        \mathrm{Ent}_{p(\cdot \mid Y_{-i})}\bigl(g\bigr)
    \Bigr],
  \]
\noindent where $ \mathrm{Ent}_p(g):= \mathbb{E}_p\bigl[g \log g\bigr]
-\mathbb{E}_p[g]\log\bigl(\mathbb{E}_p[g]\bigr)
\in [0,\infty)$ is the entropy functional of nonnegative function $g$. 
\end{definition}
To show that our stationary distribution satisfies the approximate tensorization property, we apply the following lemma, which gives sufficient conditions for a probability distribution to satisfy approximate tensorization.
\begin{lemma}{\citep[\S Theorem~4.2]{gotze2019higher}}\label{thm:3.2}
  Let $p$ be a measure with $p(y) > 0$ for all $y \in \mathcal{Y}^N$. Define
  \begin{equation}\label{eq:beta1}
    \chi \;=\; 
      \min_{1\le i\le N}\;\min_{Y\in \mathcal{Y}^N} 
        p_i\bigl(Y_i\mid Y_{-i}\bigr).
  \end{equation}
  Let $A=(a_{ij})_{i,j\le N}$ satisfy $a_{ii}=0$ for all $i$ and for $i\neq j$
  \[
    \bigl\|p_i(\cdot \mid Y_{-i}) 
         \;-\;
         p_i(\cdot \mid Y_{-i}')\bigr\|_{\mathrm{TV}}
    \;\le\;
    a_{ij},
  \]
  whenever $Y,Y'\in \mathcal{Y}^N$ differ only at the $j$-th coordinate.
  Assume moreover that 
\begin{equation}
    \|A\|_{2,2} <1.
\end{equation}
Then $p$ satisfies the approximate tensorization property 
  $\mathrm{AT}(C)$ with
  \[
    C 
    \;=\; 
    \frac{1}{\chi\,\bigl(1-\|A\|_{2, 2}\bigr)^2}.
  \]
\end{lemma}
Lemma~\ref{lem:3.1} below verifies the assumptions of Lemma~\ref{thm:3.2} for our stationary distribution $P$. In particular, Lemma~\ref{lem:3.1} implies that
\[
\|A\|_{2,2}\le 1-\alpha,
\]
so Lemma~\ref{thm:3.2} yields that $P$ satisfies the approximate tensorization property 
\[
AT\!\left(\frac{1}{\chi\alpha^2}\right),
\]
where $\chi$ depends on $(\theta,\bar{X},G,\alpha)$ through \eqref{eq:beta1}. In addition, the second statement of Lemma \ref{lem:3.1} guarantees that $\chi\geq(1-C_{\alpha})$. Therefore, $P$ satisfies the approximate tensorization property $AT(\frac{1}{(1-C_{\alpha})\alpha^2})$. Applying Lemma \ref{thm:5.3}, we conclude that there exists a constant
\(c_{\mathrm{trans}}>0\) such that \(P\) satisfies the weak transport inequality with cost
\[
f(x)=c_{trans}\|x\|_1^2.
\]
That is, for every probability measure \(Q\) on \(\{0,1\}^N\),
\[
\overline{\mathrm T}_f(Q\mid P)\le \mathbb{KL}(Q\mid P),
\]
where
\begin{equation}\label{eq:wtp}
    \overline{\mathrm{T}}_f(Q\mid P)=\inf_{\omega}
      \int_{\{0,1\}^N}
        c_{trans}\biggl\Vert
          y - \int_{\{0,1\}^N} y'\, \pi(dy'\vert y)
        \biggr\Vert_1^2
      \,P(\mathrm{d}y), 
\end{equation}
which completes the proof of the first step. Therefore, we have
\begin{equation}\label{eq:weaktransport}
    \sqrt{\frac{1}{c_{trans}}\overline{\mathrm{T}}_f(Q\mid P)}\leq \sqrt{\frac{1}{c_{trans}}\mathbb{KL}(Q\mid P)}.
\end{equation}
Given the Wasserstein 1-distance equipped with Hamming distance, and the weak transport cost defined in Eq.\ref{eq:wtp}, if we can show 
\[
W_1(P,Q) \leq \sqrt{\frac{1}{c_{trans}}\overline{T}_f(Q\mid P)},
\]
combining with Eq.(\ref{eq:weaktransport}) with $C_{trans}=\frac{1}{\sqrt{c_{trans}}}$ leads to the current proposition. To this goal, let $P, Q$ be probability measures on $\{0,1\}^N$. Define the Wasserstein 1-distance (with Hamming distance) as
\[
W_1(P, Q) = \inf_{\omega \in \Omega(P,Q)} \mathbb{E}_\omega[\|Y - Y'\|_1],
\]
where $\|y-y'\|_1 = \sum_{i=1}^{N}|y_i - y_i'|$ is the Hamming distance. Define the weak transport cost as
\[
\overline{\mathrm{T}}_f(Q\mid P) = \inf_{\omega \in \Omega(P,Q)} \int_{\{0,1\}^N}C_{trans}\|y - \mathbb{E}_\omega[Y'|Y=y]\|_1^2 P(dy).
\]
Define another weak transport cost as
\[
\overline{\mathrm{T}}'(Q\mid P) = \inf_{\omega \in \Omega(P,Q)} \int_{\{0,1\}^N}C_{trans}^{1/2}\|y - \mathbb{E}_\omega[Y'|Y=y]\|_1 P(dy).
\]
By Jensen's inequality and the square root function is concave, we have
\begin{equation}\label{eq:tq'}
 \begin{split}
     \overline{\mathrm{T}}'(Q\mid P)&=\inf_{\omega \in \Omega(P,Q)} \int_{\{0,1\}^N}\sqrt{C_{trans}\|y - \mathbb{E}_\omega[Y'|Y=y]\|_1^2} P(dy)\\
     &\leq\inf_{\omega \in \Omega(P,Q)} \left[\int_{\{0,1\}^N}C_{trans}\|y - \mathbb{E}_\omega[Y'|Y=y]\|_1^2 P(dy)\right]^{1/2}\\
&=\sqrt{\overline{\mathrm{T}}_f(Q\mid P)}.
 \end{split}
\end{equation}
To show
\begin{equation}\label{eq:teq}
    W_1(P, Q) = \frac{1}{\sqrt{C_{trans}}}\overline{\mathrm{T}}'(Q\mid P),
\end{equation}
consider each component in $\|y - \mathbb{E}_\omega[Y'|Y=y]\|_1$ separately. Given coupling $\pi$ and conditioning on $Y=y$, we have
\[
\mathbb{E}_\omega[Y_i'|Y=y] = \mathbb{P}_\omega(Y_i'=1|Y=y) =: p_i(y).
\]
Thus,
\[
|y_i - p_i(y)| = \mathbb{P}_\omega(Y_i' \neq y_i|Y=y).
\]
Summing over units, we get:
\[
\|y - \mathbb{E}_\omega[Y'|Y=y]\|_1 = \sum_{i=1}^{N}\mathbb{P}_\omega(Y_i' \neq y_i|Y=y) = \mathbb{E}_\omega[\|y - Y'\|_1|Y=y].
\]
Integrating w.r.t. $P(dy)$, we have the exact equality:
\[
\int_{\{0,1\}^N}\|y - \mathbb{E}_\omega[Y'|Y=y]\|_1 P(dy) = \mathbb{E}_\omega[\|Y - Y'\|_1].
\]
Since this holds for any coupling $\pi$, we have equality of objective functions for each coupling. Thus, the infimum over all couplings must coincide.
This establishes the desired equivalence. Combining Eq.\ref{eq:tq'} and Eq.\ref{eq:teq}, we finish the proof of our claim.

\end{proof}
\subsection{Lemma \ref{lem:3.1}}
\begin{lemma}\label{lem:3.1}
  Let $P$ be the stationary distribution defined in Eq.\ref{eqjointdis} and the utility function specification of Eq.\ref{eq:alpha} and \ref{eq:beta}. Then, a coupling matrix is given by $J$, with each element $J_{ij}=  \frac{A_N}{4}m_{ij}G_{ij}\big(\vert\theta_5\vert+\vert\theta_6\vert \big) $. Under Assumption \ref{ass:boundspill}, there exists $\alpha \in (0,1)$ such that
\begin{equation}\label{eq:2if}
    \Vert J\Vert_{2,2}\leq \Vert J\Vert_{\infty,\infty}\leq
    1 - \alpha
    \quad
    \text{holds for }P.
\end{equation}
  Moreover, 
  \[
    P_i\bigl(\,\cdot \mid Y_{-i}\bigr)
    \;\in\;
    \bigl(1-C_{\alpha},\; 
           C_{\alpha}\bigr)
    \quad
    \text{for some } 
      C_{\alpha}\in(0,1)
      \text{ depending only on }
      \alpha, \theta, \widebar{N}, \widebar{m}, \text{ and } \bar{X}
  \]
  uniformly in $i, N,$ and $Y_{-i}$.
\end{lemma}
\begin{proof}
\begin{equation}
    P_i\bigl(1 \mid Y_{-i}\bigr)=\frac{\exp\big(\Phi(y_i=1,Y_{-i})\big)}{\exp\big(\Phi(y_i=1,Y_{-i})\big)+\exp\big(\Phi(y_i=0,Y_{-i})\big)},
\end{equation}
\begin{equation}
    P_i\bigl(0 \mid Y_{-i}\bigr)=\frac{\exp\big(\Phi(y_i=0,Y_{-i})\big)}{\exp\big(\Phi(y_i=1,Y_{-i})\big)+\exp\big(\Phi(y_i=0,Y_{-i})\big)},
\end{equation}
where
\begin{equation}
\small
    \Phi(y_i=1,Y_{-i}) = \alpha_i+\sum_{j\neq i} \alpha_j Y_{j}+\frac{A_N}{2}\sum_{\ell=1}^N \sum_{j=1}^N m_{\ell j}G_{\ell j} Y_{\ell}Y_j(\theta_5+\theta_6 d_\ell d_j),
\end{equation}
\begin{equation}
    \Phi(y_i=0,Y_{-i}) =\sum_{j\neq i} \alpha_j Y_{j}+\frac{A_N}{2}\sum_{\ell\neq i} \sum_{j\neq i} m_{\ell j}G_{\ell j} Y_{\ell}Y_j(\theta_5+\theta_6 d_\ell d_j),
\end{equation}
where $\alpha_i=\theta_0 +\theta_1  d_i+X_i'(\theta_2 +\theta_3d_i)+A_N\sum_{j=1}^N \theta_4 m_{ij}G_{ij}d_j$. Recall $\Lambda(x)=\frac{\exp(x)}{1+\exp(x)}$ as a sigmoid function of $x\in\mathbb{R}$, then we have
\begin{equation}\label{eq:sig}
\small
    P_i\bigl(1 \mid Y_{-i}\bigr) = \Lambda\left(\alpha_i+A_N\sum_{j=1}^N m_{ij}G_{ij} Y_j(\theta_5+\theta_6 d_id_j)\right).
\end{equation}
\begin{equation}
\small
    P_i\bigl(0 \mid Y_{-i}\bigr) = 1-P_i\bigl(1 \mid Y_{-i}\bigr).
\end{equation}
Let $i\neq k$ be fixed and $z,y\in\mathcal{Y}$ be such that $y$ and $z$ differ in the $k$-th coordinate only, i.e., $y=T_{k}z$. Therefore,
\begin{equation}
    d_{TV}(P_i\bigl(\,\cdot \mid z_{-i}\bigr),P_i\bigl(\,\cdot \mid y_{-i}\bigr)) = \vert P_i\bigl(1 \mid z_{-i}\bigr)-P_i\bigl(1 \mid y_{-i}\bigr)\vert=\vert \Lambda(g_i(y)) - \Lambda(g_i(T_ky)) \vert,
\end{equation}
where $g_i(Y)=\alpha_i+A_N\sum_{j=1}^N m_{ij}G_{ij} Y_j(\theta_5+\theta_6 d_id_j)$, and hence,
\begin{equation}
  d_{TV}(P_i\bigl(\,\cdot \mid z_{-i}\bigr),P_i\bigl(\,\cdot \mid y_{-i}\bigr))\leq  \frac{1}{4}\vert g_i(y)- g_i(T_ky)\vert  \leq \frac{1}{4}\big\vert A_Nm_{ik}G_{ik}(\theta_5+\theta_6)\big\vert. 
\end{equation}
Thus, a coupling matrix can be given by $J$, with each element $J_{ij}=  \frac{A_N}{4}m_{ij}G_{ij}\big(\vert\theta_5\vert+\vert\theta_6\vert \big) $. First inequality in Eq.\ref{eq:2if} holds by
\begin{equation}
    \Vert J\Vert_{2,2}\leq\sqrt{\Vert J\Vert_{\infty,\infty}\Vert J^T\Vert_{\infty,\infty}}=\Vert J\Vert_{\infty,\infty},
\end{equation}
where the eigenvalue satisfies $\vert \lambda_i(JJ^T) \vert\leq \Vert J J^T\Vert\leq \Vert J\Vert \Vert J^T \Vert  $ for any operator norm, and $J$ is symmetric. Under Assumption \ref{ass:boundspill}, there exists $\alpha \in (0 ,1)$ such that it holds
\begin{equation} 
\Vert J\Vert_{\infty,\infty}=\frac{A_N}{4}\big(\vert\theta_5\vert+\vert\theta_6\vert \big)\max_{i=1,...,N}\sum_{j=1}^Nm_{ij}G_{ij}\leq 1-\alpha.
\end{equation}
The second statement follows by using Eq.\ref{eq:sig},
\begin{equation}
    P_i(1\vert Y_{-i})\leq \Lambda(\vert \theta_0\vert+\vert \theta_1\vert+\overline{X}'\vert \theta_2\vert+\overline{X}'\vert \theta_3\vert+A_N\overline{N}\overline{m}\vert\theta_4\vert+4(1-\alpha))\eqqcolon C_{\alpha}.
\end{equation}
Therefore, $P_i(0\mid Y_{-i})\geq 1-C_{\alpha}.$ In addition, by the symmetry of the logit distribution around the origin, $P_i(1\vert Y_{-i}) \geq 1- C_{\alpha}$ also follows.
\end{proof}



\begin{singlespace}
\bibliographystyle{ecta}  
\bibliography{ref}
\end{singlespace}

\newpage

\begin{center}\LARGE {Supplement to ``Individualized Treatment Allocations in Sequential Network Games''}
\end{center}

\bigskip
\section{Additionl Table}
\begin{table}[h]
\begingroup
\setlength{\tabcolsep}{6pt} 
\renewcommand{\arraystretch}{0.88}
 \begin{adjustwidth}{0cm}{}
 \linespread{1}
\footnotesize
\centering 
\begin{threeparttable}
 \begin{tabular}{@{}lcccccc@{}}
   \hline
   \toprule
    & \multicolumn{6}{c}{\begin{tabular}[c]{@{}c@{}}\textit{Network Size}\end{tabular}}\\
  \cmidrule(lr){2-7}
  \textbf{\textit{Allocation Rule}} & \begin{tabular}[c]{@{}c@{}}$N=5$\end{tabular} & \begin{tabular}[c]{@{}c@{}}$N=7$\end{tabular} & \begin{tabular}[c]{@{}c@{}}$N=9$\end{tabular} & \begin{tabular}[c]{@{}c@{}}$N=11$\end{tabular} & \begin{tabular}[c]{@{}c@{}}$N=13$\end{tabular} & \begin{tabular}[c]{@{}c@{}}$N=15$\end{tabular}\\
  \midrule
  \textit{Village $1$}\\
  &\\
  \textbf{Brute force}&$0.28$ &$0.27$ &$0.27$ &$0.27$ &$0.27$ &$0.28$ \\
  &  ($<0.01$)& ($<0.01$)&($<0.01$)&($<0.01$)& ($<0.01$)&($<0.01$)\\
  \textbf{Brute force with var. approx.}&$0.28$ &$0.27$ &$0.27$ &$0.27$ &$0.27$ &$0.28$ \\
  &  ($<0.01$)& ($<0.01$)&($<0.01$)&($<0.01$)& ($<0.01$)&($<0.01$)\\
  \textbf{Greedy with var. approx.}&$0.28$ &$0.27$ &$0.27$ &$0.27$ &$0.27$ &$0.28$ \\
  &  ($<0.01$)& ($<0.01$)&($<0.01$)&($<0.01$)& ($<0.01$)&($<0.01$)\\
  \textbf{Random}& $0.24 $ &$0.23$ &$0.24$ &$0.23$ &$0.24$ &$0.24$ \\
  &  ($<0.01$)& ($<0.01$)&($<0.01$)&($<0.01$)& ($<0.01$)&($<0.01$)\\
  \textbf{Centrality}& $0.25$ &$0.23$ &$0.24$ &$0.23$ &$0.24$ &$0.24$ \\
  &  ($<0.01$)& ($<0.01$)&($<0.01$)&($<0.01$)& ($<0.01$)&($<0.01$)\\
   \textit{Village $4$}\\
  &\\
  \textbf{Brute force}&$0.13$ &$0.13$ &$0.13$ &$0.12$ &$0.13$ &$0.14$ \\
  &  ($<0.01$)& ($<0.01$)&($<0.01$)&($<0.01$)& ($<0.01$)&($<0.01$)\\
   \textbf{Brute force with var. approx.}&$0.13$ &$0.13$ &$0.13$ &$0.12$ &$0.13$ &$0.14$ \\
  &  ($<0.01$)& ($<0.01$)&($<0.01$)&($<0.01$)& ($<0.01$)&($<0.01$)\\
  \textbf{Greedy with var. approx.}&$0.13$ &$0.13$ &$0.13$ &$0.12$ &$0.13$ &$0.14$ \\
  &  ($<0.01$)& ($<0.01$)&($<0.01$)&($<0.01$)& ($<0.01$)&($<0.01$)\\
  \textbf{Random}& $0.09$ &$0.08$ &$0.08$ &$0.08$ &$0.08$ &$0.09$ \\
  &  ($<0.01$)& ($<0.01$)&($<0.01$)&($<0.01$)& ($<0.01$)&($<0.01$)\\
  \textbf{Centrality}& $0.08$ &$0.08$ &$0.08$ &$0.08$ &$0.08$ &$0.08$ \\
  &  ($<0.01$)& ($<0.01$)&($<0.01$)&($<0.01$)& ($<0.01$)&($<0.01$)\\
\bottomrule
   \end{tabular}
   \end{threeparttable}
   \caption{Comparison of Four Allocation Methods}
\label{table brute}
   \end{adjustwidth}
   \endgroup
\end{table}

\section{Lemma and Proposition}\label{appendixA}
\subsection{Preliminary Lemma}
In this section, we collect various lemmas that we use to prove our main results.

\begin{lemma}\label{lemmatau}
For any $\tau\in\mathbb{R}_+$, there is a finite set of $N\times 1$ vectors $\mathcal{M}(\tau)$ such that
\begin{equation}
    \vert\mathcal{M}(\tau)\vert\leq 2^N,
\end{equation}
and for any $N\times 1$ vector $Y$ with entries in $\{0,1\}$, there exists a $M\in \mathcal{M}(\tau)$ such that
\begin{equation}
    \sum_{i}(Y_i-M_i)^2\leq\tau^2.
\end{equation}
\end{lemma}

\begin{proof}
    Since all the entries of $Y$ are binary, $0$ or $1$, $Y$ must be a vertex of the N-dimensional unit hypercube $[0,1]^N$. 
    As such, we let $\mathcal{M}$ be the collection of all vertices of the N-dimensional unit hypercube. 
    For any $Y$, we can always find an element in $\mathcal{M}$ such that
    \begin{equation}
        \sum_{i}(Y_i-M_i)^2=0\leq\tau^2.
    \end{equation}
    Then $\vert \mathcal{M}(\tau)\vert=2^N$.
\end{proof}

\begin{lemma}\label{lemmamonoto}
  Suppose Assumptions \ref{assumpundir} to \ref{assptse} hold. 
  Let $R \subset \mathcal{N}$ be a treatment allocation set, $R = \{ i \in \mathcal{N} : d_i = 1 \}$ such that $\mathcal{N}\setminus R \neq \emptyset$. 
  Given $k\in\mathcal{N}\setminus R$, let $\tilde{\mu}_i$, $i=1,\dots,N$, be a solution of the first-order conditions of Eq.\ref{eqmui} when the treatment allocation set is $R \cup \{k\}$ and $\breve{\mu}_i$, $i=1,\dots,N$, be a solution of Eq.\ref{eqmui} when the treatment allocation set is $R$, i.e., 
    \begin{equation}
 \Tilde{\mu}_i=\Lambda\big(\theta_0+\theta_1d_i +X_i'(\theta_2 +\theta_3d_i)+A_N\theta_5 \sum_{j\neq i}m_{ij}G_{ij}\Tilde{\mu}_j+A_N\sum_{\substack{j\neq i\\j\in R}}m_{ij}G_{ij}(\theta_4+\theta_6d_i\Tilde{\mu}_j)+\mathcal{M}_{ik}(\Tilde{\mu})\big), 
    \end{equation}
    \begin{equation}
       \breve{\mu}_i=\Lambda\big(\theta_0+\theta_1 d_i+X_i'(\theta_2 +\theta_3d_i)+A_N \theta_5\sum_{j\neq i}m_{ij}G_{ij}\breve{\mu}_j+A_N\sum_{\substack{j\neq i\\j\in R}}m_{ij}G_{ij}(\theta_4+\theta_6d_i\breve{\mu}_j)\big),
    \end{equation}
where $\mathcal{M}_{ik}: [0,1]^N \to \mathbb{R}_+$ is defined by 
\begin{equation}
    \mathcal{M}_{ik}(\mu)=A_Nm_{ik}G_{ik}(\theta_4+\theta_6d_i\mu_k)+\mathbbm{1}\{i=k\}(X_i'\theta_3+A_N\sum_{j\in R}\theta_6m_{ij}G_{ij}\mu_j).
\end{equation}
Then, $\Tilde{\mu}_i\geq \breve{\mu}_i$ holds for all $i\in \mathcal{N}$ at any $\mathcal{X}\in\mathbb{R}^{N\times k}$ and $G\in\{0,1\}^{N\times N}$.
\end{lemma}

\begin{proof}
    Let us define $\Tilde{\mu}$ as a vector with elements $\{\Tilde{\mu}_i\}_{i=1}^N$ and $\breve{\mu}$ as a vector with elements $\{\breve{\mu}_i\}_{i=1}^N$. By Assumption \ref{assptse} (i), $\mathcal{M}_{ik} \geq 0$.       
We define 
\begin{equation}
    \Tilde{\mu}_i^1=\Lambda\big(\theta_0+\theta_1d_i +X_i'(\theta_2 +\theta_3d_i)+A_N\theta_5 \sum_{j\neq i}m_{ij}G_{ij}\breve{\mu}_j+A_N\sum_{\substack{j\neq i\\j\in R}}m_{ij}G_{ij}(\theta_4+\theta_6d_i\breve{\mu}_j)+\mathcal{M}_{ik}(\breve{\mu})\big),
\end{equation}
for all $i\in\mathcal{N}$, where $\Tilde{\mu}^1_i\geq \breve{\mu}_i$. 
Then, we use $\Tilde{\mu}^1$ to generate $\Tilde{\mu}^2$,
\begin{equation}
    \Tilde{\mu}_i^2=\Lambda\big(\theta_0+\theta_1d_i +X_i'(\theta_2 +\theta_3d_i)+A_N\theta_5 \sum_{j\neq i}m_{ij}G_{ij}\Tilde{\mu}_j^1+A_N\sum_{\substack{j\neq i\\j\in R}}m_{ij}G_{ij}(\theta_4+\theta_6d_i\Tilde{\mu}_j^1)+\mathcal{M}_{ik}(\Tilde{\mu}^1)\big),
\end{equation}
for all $i\in\mathcal{N}$, where $\Tilde{\mu}^2_i\geq \Tilde{\mu}_i^1$. 
We iterate the above process until it converges.
As shown in the proof of Proposition \ref{prounimax}, this iteration is a contraction mapping, which guarantees convergence to $\Tilde{\mu}$ by Banach Fixed-Point Theorem.\footnote{Banach Fixed-Point Theorem states that if $(X,d)$ is a non-empty complete metric space with a contraction mapping $T: X\rightarrow X$, then $T$ admits a unique fixed-point $x^*$.}
\end{proof}

\begin{lemma}\label{lemmabound}
    Under Assumptions \ref{assumpundir} to \ref{assptse}, an upper bound on the curvature $\xi^{up}$ and a lower bound on the submodularity ratio $\gamma^{low}$ satisfy
    \begin{equation}
        \xi\leq \xi^{up}<1,
    \end{equation}
    \begin{equation}
        \gamma\geq\gamma^{low}>0.
    \end{equation}

\end{lemma}
\begin{proof}    
We first rewrite our objective function in terms of set function ($\mathcal{D}=\{i\in\mathcal{N}:d_i=1\}$)
\begin{equation}
            \begin{split}
                \Tilde{W}(\mathcal{D})&= \sum_{i\in \mathcal{D}}\Lambda\big[\theta_0+\theta_1 +X_i'(\theta_2+\theta_3) +A_N\theta_5 \sum_{\substack{j\neq i\\j\in\mathcal{N}}}m_{ij}G_{ij}\Tilde{\mu}_j+A_N\sum_{\substack{j\neq i\\j\in \mathcal{D}}}m_{ij}G_{ij}(\theta_4+\theta_6\Tilde{\mu}_j)\big]\\&\quad+\sum_{k\in\mathcal{N}\setminus\mathcal{D}}\Lambda\big[\theta_0+X_k'\theta_2+A_N\theta_4\sum_{\ell\in\mathcal{D}}m_{k\ell}G_{k\ell}+A_N\theta_5\sum_{\substack{\ell\neq k\\\ell\in\mathcal{N}}}m_{k\ell}G_{k\ell}\Tilde{\mu}_{\ell}\big].
            \end{split}
        \end{equation}  
\noindent\textbf{Curvature}:\\
\noindent The curvature is defined as the smallest value of $\xi$ such that
\begin{equation}
\Tilde{W}(R\cup \{k\})-\Tilde{W}(R)\geq (1-\xi)[\Tilde{W}(S\cup \{k\})-\Tilde{W}(S)]\quad   \forall S\subseteq R\subset \mathcal{N}, \forall k\in \mathcal{N}\setminus R.
\end{equation}
As a consequence,
\begin{equation}\label{eq2}
        \xi =\max_{S\subseteq R\subset\mathcal{N},k\in\mathcal{N}\setminus R} 1-\frac{\Tilde{W}(R\cup \{k\})-\Tilde{W}(R)}{\Tilde{W}(S\cup \{k\})-\Tilde{W}(S)}.
\end{equation}
Define $\Tilde{W}_i(\mathcal{D})$ as
\begin{equation}
    \begin{split}
        \Tilde{W}_i(\mathcal{D}) &\coloneqq  \mathbbm{1}_{\{i\in \mathcal{D}\}}\Lambda\big[\theta_0+\theta_1 +X_i'(\theta_2+\theta_3) +A_N\theta_5 \sum_{\substack{j\neq i\\j\in\mathcal{N}}}m_{ij}G_{ij}\Tilde{\mu}_j+A_N\sum_{\substack{j\neq i\\j\in \mathcal{D}}}m_{ij}G_{ij}(\theta_4+\theta_6\Tilde{\mu}_j)\big]\\&\quad+\mathbbm{1}_{\{i\in\mathcal{N}\setminus\mathcal{D}\}}\Lambda\big[\theta_0+X_i'\theta_2+A_N\theta_4\sum_{\ell\in\mathcal{D}}m_{i\ell}G_{i\ell}+A_N\theta_5\sum_{\substack{\ell\neq i\\\ell\in\mathcal{N}}}m_{i\ell}G_{i\ell}\Tilde{\mu}_{\ell}\big].
    \end{split}
\end{equation}
We can upper bound the denominator in Eq.\ref{eq2} by
\begin{equation}
    \begin{split}
        \Tilde{W}(S\cup \{k\})-\Tilde{W}(S)&=\sum_{i=1}^N \Tilde{W}_i(S\cup \{k\})-\Tilde{W}_i(S)\leq N,
    \end{split}
\end{equation}
We derive a lower bound of the numerator in Eq.\ref{eq2} in what follows. Given $R \subset \mathcal{N}$ and $k \in \mathcal{N} \setminus R$, let $\Tilde{\mu}$ denote the solution of Eq.\ref{eqva} with $d_i=1$ for $i\in R\cup \{k\}$ and $\breve{\mu}$ denote the solution of Eq.\ref{eqva} with $d_i=1$ for $i\in R$. 
To derive a lower bound for $\tilde{W}(R\cup \{k\})-\tilde{W}(R)$, consider

\begin{equation}
    \overline{\phi}_i(\mu)\coloneqq\theta_0+\theta_1 +X_i'(\theta_2+\theta_3) +A_N\theta_5 \sum_{j\neq i}m_{ij}G_{ij}\mu_j+A_N\sum_{\substack{j\neq i\\j\in R}}m_{ij}G_{ij}(\theta_4+\theta_6\mu_j)+A_Nm_{ik}G_{ik}(\theta_4+\theta_6\mu_k),
\end{equation}

\begin{equation}
    \underline{\phi}_i(\mu)\coloneqq\theta_0+\theta_1 +X_i'(\theta_2+\theta_3)+A_N\theta_5 \sum_{j\neq i}m_{ij}G_{ij}\mu_j+A_N\sum_{\substack{j\neq i\\j\in R}}m_{ij}G_{ij}(\theta_4+\theta_6\mu_j),
\end{equation}

\begin{equation}
    \overline{\varphi}_c(\mu)\coloneqq\theta_0+X_c'\theta_2+A_N\theta_5\sum_{b\neq c}G_{cb}m_{bc}\mu_b+A_N\theta_4\sum_{\substack{z\in R}}m_{cz}G_{cz}+A_N\theta_4m_{ck}G_{ck},
\end{equation}

\begin{equation}
    \underline{\varphi}_c(\mu)\coloneqq\theta_0+X_c'\theta_2+A_N\theta_5\sum_{b\neq c}G_{cb}m_{bc}\mu_b+A_N\theta_4\sum_{\substack{z\in R}}m_{cz}G_{cz},
\end{equation}

\begin{equation}
    \overline{\psi}_k(\mu)\coloneqq\theta_0+\theta_1+X_k'(\theta_2+\theta_3) +A_N\theta_5\sum_{e\neq k}m_{ke}G_{ke}\mu_{e}+A_N\sum_{\substack{l\in R}}m_{kl}G_{kl}(\theta_4+\theta_6\mu_l),
\end{equation}

\begin{equation}
    \underline{\psi}_k(\mu)\coloneqq\theta_0+X_k'\theta_2+A_N\theta_5\sum_{e\neq k}m_{ke}G_{ke}\mu_{e}+A_N\theta_4\sum_{l\in R}m_{kl}G_{kl}.
\end{equation}
Then, we have
\begin{equation}\label{eqapp91}
    \begin{split}
        &\quad \Tilde{W}(R\cup \{k\})-\Tilde{W}(R)\\&=\sum_{i\in R}\left[\Lambda\left(\overline{\phi}_i(\Tilde{\mu})\right)-\Lambda\left(\underline{\phi}_i(\breve{\mu})\right)\right]
        +\sum_{c\in \mathcal{N}\setminus R\cup\{k\}}\left[\Lambda \Big(\overline{\varphi}_c(\Tilde{\mu})\Big)-\Lambda\left(\underline{\varphi}_c(\breve{\mu})\right)\right]+\Lambda\left(\overline{\psi}_k(\Tilde{\mu}) \right)-\Lambda\left(\underline{\psi}_k(\breve{\mu})\right).
    \end{split}
\end{equation}

By Lemma \ref{lemmamonoto}, we can bound the above equation from below by replacing $\breve{\mu}_i  $ with $\Tilde{\mu}_i$ for all $i\in\mathcal{N}$. 
Then, Eq.\ref{eqapp91} is bounded by

\begin{equation}
    \begin{split}
         &\Tilde{W}(R\cup \{k\})-\Tilde{W}(R)\\
         \geq &\sum_{i\in R}\left[\Lambda\left(\overline{\phi}_i(\Tilde{\mu})\right)-\Lambda\left(\underline{\phi}_i(\Tilde{\mu})\right)\right]
        +\sum_{c\in \mathcal{N}\setminus R\cup\{k\}}\left[\Lambda \Big(\overline{\varphi}_c(\Tilde{\mu})\Big)-\Lambda\left(\underline{\varphi}_c(\Tilde{\mu})\right)\right]+\Lambda\left(\overline{\psi}_k(\Tilde{\mu}) \right)-\Lambda\left(\underline{\psi}_k(\Tilde{\mu})\right).
           \end{split}
\end{equation}

Using the mean value theorem, and letting $\phi_i\in\left(\underline{\phi}_i(\Tilde{\mu}),\overline{\phi}_i(\Tilde{\mu})\right)$, $\varphi_c\in\left(\underline{\varphi}_c(\Tilde{\mu}),\overline{\varphi}_c(\Tilde{\mu})\right)$, $\psi_k\in\left(\underline{\psi}_k(\Tilde{\mu}),\overline{\psi}_k(\Tilde{\mu})\right)$, we have
\begin{equation}\label{eqrk}
    \begin{split}
        &\quad \Tilde{W}(R\cup \{k\})-\Tilde{W}(R)\\&\geq \sum_{i\in R}\Lambda'(\phi_i)A_N(\theta_4+\theta_6\Tilde{\mu}_k)m_{ik}G_{ik}+\sum_{c\in \mathcal{N}\setminus R\cup\{k\}}\Lambda'(\varphi_c)A_N\theta_4m_{ck}G_{ck}\\&\quad+\Lambda'(\psi_k)(\theta_1+X_k'\theta_3+\theta_6A_N\sum_{l\in R}m_{kl}G_{kl}\Tilde{\mu}_l)\\
        &\geq \sum_{i\in R}\Lambda'(\phi_i)A_N\theta_4m_{ik}G_{ik}+\sum_{c\in \mathcal{N}\setminus R\cup\{k\}}\Lambda'(\varphi_c)A_N\theta_4m_{ck}G_{ck}+\Lambda'(\psi_k)(\theta_1+X_k'\theta_3+\theta_6A_N\sum_{l\in R}m_{kl}G_{kl}\Tilde{\mu}_l)\\
        &\geq \sum_{i\in R}\underline{\Lambda'(\phi_i)}A_N\theta_4m_{ik}G_{ik}+\sum_{c\in \mathcal{N}\setminus R\cup\{k\}}\underline{\Lambda'(\varphi_c)}A_N\theta_4m_{ck}G_{ck}+\underline{\Lambda'(\psi_k)}(\theta_1+X_k'\theta_3+\theta_6A_N\sum_{l\in R}m_{kl}G_{kl}\Tilde{\mu}_l)\\
        &\geq \underline{\Lambda'}\cdot \big(\sum_{i\neq k}A_N\theta_4m_{ik}G_{ik}+\theta_1+X_k'\theta_3+\theta_6A_N\sum_{l\in R}m_{kl}G_{kl}\Tilde{\mu}_l\big)\\
        &\geq \underline{\Lambda'}\cdot \big(\sum_{i\neq k}A_N\theta_4m_{ik}G_{ik}+\theta_1+X_k'\theta_3\big),
    \end{split}
\end{equation}
where 
\begin{equation}
    \underline{\Lambda'(\phi_i)}=\min\{\Lambda'(\underline{\phi}_i),\Lambda'(\overline{\phi}_i)\},
\end{equation}
\begin{equation}
    \underline{\Lambda'(\varphi_c)}=\min\{\Lambda'(\underline{\varphi}_c),\Lambda'(\overline{\varphi}_c)\},
\end{equation}
\begin{equation}
    \underline{\Lambda'(\psi_k)}=\min\{\Lambda'(\underline{\psi}_k),\Lambda'(\overline{\psi}_k)\},
\end{equation}
and noting that $\phi_i,\varphi_c,\psi_k\in\left[\theta_0+ \underline{X'\theta_2},\;\theta_0+\theta_1+\overline{X'\theta_2}+\overline{X'\theta_3}+A_N(\theta_4+\theta_5+\theta_6) \overline{m}\widebar{N}\right]$ for all $i,k,c\in\mathcal{N}$ with $\underline{X'\theta_2}=\min_{i\in \mathcal{N}} X_i'\theta_2$, $\overline{X'\theta_2}=\max_{i\in \mathcal{N}} X_i'\theta_2$, $ \overline{X'\theta_3}=\max_{i\in\mathcal{N}} X_i'\theta_3$,
\begin{equation}
    \underline{\Lambda'}=\min\{\Lambda'\big(\theta_0+ \underline{X'\theta_2}\big),\Lambda'\big(\theta_0+\theta_1+\overline{X'\theta_2}+\overline{X'\theta_3}+A_N(\theta_4+\theta_5+\theta_6) \overline{m}\widebar{N}\big)\}.
\end{equation}
Then,
\begin{equation}
    \xi^{up}=\max_{k\in\mathcal{N}} 1-\frac{1}{N}\underline{\Lambda'}\cdot\big(\sum_{i\neq k}A_N\theta_4m_{ik}G_{ik}+\theta_1+X_k'\theta_3\big).
\end{equation}
Under Assumptions \ref{assdp} and \ref{assptse}, for any $k\in\mathcal{N}$,
\begin{equation}\label{eq:weakcondi}
    \begin{split}  \sum_{i\neq k}A_N\theta_4m_{ik}G_{ik}+\theta_1+X_k'\theta_3>0.
    \end{split}
\end{equation}
In addition, we know that $\Lambda(x)$ is a logistic function, and so
\begin{equation}
    \begin{split}
        &\Lambda'\big(\theta_0+ \underline{X'\theta_2}\big)\in (0,0.25],\\
        &\Lambda'\big(\theta_0+\theta_1+\overline{X'\theta_2}+\overline{X'\theta_3}+A_N(\theta_3+\theta_4+\theta_5) \overline{m}\widebar{N}\big)\in(0,0.25].\\
    \end{split}
\end{equation}
Hence,
\begin{equation}
    \begin{split}
    \underline{\Lambda'}\in(0,0.25).
    \end{split}
    \end{equation}
By assuming that $\sum_{i\neq k}A_N\theta_4m_{ik}G_{ik}+\theta_1+X_k'\theta_3\leq4N$, we obtain
    \begin{equation}
        \begin{split}
           \underline{\Lambda'}\cdot\big(\sum_{i\neq k}A_N\theta_4m_{ik}G_{ik}+\theta_1+X_k'\theta_3\big)<N.
        \end{split}
    \end{equation}
We conclude that $\xi^{up}<1$.

\hfil\\
\noindent\textbf{Submodularity Ratio}:

The submodularity ratio of a non-negative set function is the largest $\gamma$ such that
\begin{equation}
    \sum_{k\in R\setminus S} [\Tilde{W}(S\cup \{k\})-\Tilde{W}(S)]\geq \gamma [\Tilde{W}(S\cup R)-\Tilde{W}(S)], \quad\forall S,R\subseteq\mathcal{N}.
\end{equation}
As a consequence,
\begin{equation}\label{eq20}
    \begin{split}
        \gamma&= \min_{S\neq R} \frac{\sum_{k\in R\setminus S} [\Tilde{W}(S\cup \{k\})-\Tilde{W}(S)]}{\Tilde{W}(S\cup R)-\Tilde{W}(S)}\\
        &= \min_{S\neq R}\frac{\sum_{k\in R\setminus S} \sum_i[\Tilde{W}_i(S\cup \{k\})-\Tilde{W}_i(S)]}{\sum_i [\Tilde{W}_i(S\cup R)-\Tilde{W}_i(S)]}
    \end{split}
\end{equation}
We can upper bound the denominator in Eq.\ref{eq20} by
\begin{equation}
        \sum_{i=1}^N [\Tilde{W}_i(S\cup R)-\Tilde{W}_i(S)]\leq N.
\end{equation}
We first rewrite $ \sum_{k\in R\setminus S} [\Tilde{W}(S\cup \{k\})-\Tilde{W}(S)]$ with a view to deriving a lower bound. 
\begin{equation}
    \begin{split}
        &\quad\sum_{k\in R\setminus S} \sum_i[\Tilde{W}_i(S\cup \{k\})-\Tilde{W}_i(S)]\\&=\sum_{k\in R\setminus S} \Bigg[\sum_{i\in S}\Big[\Lambda\big(\theta_0+\theta_1 +X_i'(\theta_2+\theta_3) +A_N\theta_5 \sum_{j\neq i}m_{ij}G_{ij}\Tilde{\mu}_j+A_N\sum_{\substack{j\neq i\\j\in S}}m_{ij}G_{ij}(\theta_4+\theta_6\Tilde{\mu}_j)\\
        &\quad+ A_Nm_{ik}G_{ik}(\theta_4+\theta_6\Tilde{\mu}_k)\big)-\Lambda\big(\theta_0+\theta_1 +X_i'(\theta_2+\theta_3)+A_N\theta_5 \sum_{j\neq i}m_{ij}G_{ij}\Tilde{\mu}_j'+A_N\sum_{\substack{j\neq i\\j\in S}}m_{ij}G_{ij}(\theta_4+\theta_6\Tilde{\mu}_j')\big)\Big]\\
        &\quad+\sum_{c\in \mathcal{N}\setminus S\cup\{k\}}\Big[\Lambda\big(\theta_0+X_c'\theta_2+A_N\theta_5\sum_{b\neq c}G_{cb}m_{bc}\Tilde{\mu}_b+A_N\theta_4\sum_{\substack{z\in S}}m_{cz}G_{cz}+A_N\theta_4m_{ck}G_{ck}\big)\\ &\quad-\Lambda(\theta_0+X_c'\theta_2+A_N\theta_5\sum_{b\neq c}G_{cb}m_{bc}\Tilde{\mu}_b'+A_N\theta_4\sum_{\substack{z\in S}}m_{cz}G_{cz})\Big]\\
        &\quad+\Lambda(\theta_0+\theta_1+X_k'(\theta_2+\theta_3) +A_N\theta_5\sum_{e\neq k}m_{ke}G_{ke}\Tilde{\mu}_{e}+A_N\sum_{\substack{l\in S}}m_{kl}G_{kl}(\theta_4+\theta_6\Tilde{\mu}_l))\\
        &\quad-\Lambda(\theta_0+X_k'\theta_2+A_N\theta_5\sum_{e\neq k}m_{ke}G_{ke}\Tilde{\mu}_{e}'+A_N\theta_4\sum_{l\in S}m_{kl}G_{kl})\Bigg].
    \end{split}
\end{equation}
We can lower bound the numerator in Eq.\ref{eq20} by
\begin{equation}
    \begin{split}
        \sum_{k\in R\setminus S} \sum_i[\Tilde{W}_i(S\cup \{k\})-\Tilde{W}_i(S)]&\geq  \sum_{k\in R\setminus S} \underline{\Lambda'}\cdot\big(A_N\theta_4\sum_{l\in\mathcal{S}}G_{kl}m_{kl}+\theta_1+X_k'\theta_3\big)\\
        &(\text{By Eq.\ref{eqrk}})\\
        &\geq \underline{\Lambda'}\cdot\big(A_N\theta_4\sum_{l\in\mathcal{S}}G_{kl}m_{kl}+\theta_1+X_k'\theta_3\big)
    \end{split}
\end{equation}
Defining:
\begin{equation}
\gamma^{low} =    \min_{k\in\mathcal{N},\mathcal{S}\subseteq\mathcal{N}} \frac{1}{N}\underline{\Lambda'}\cdot\big(A_N\theta_4\sum_{l\in\mathcal{S}}G_{kl}m_{kl}+\theta_1+X_k'\theta_3\big),
\end{equation}
and the submodularity ratio can be bounded from below by
\begin{equation}
   \gamma\geq  \gamma^{low}=1-\xi^{up}>0.
\end{equation}
\end{proof}

\subsection{Proof of Proposition \ref{propoten}}\label{appenlemma1} 
\begin{proof}
    A potential function is a function $\Phi:Y\rightarrow\mathbb{R}$ such that
     \begin{equation}
        \Phi(y_i=1,y_{-i},\mathcal{X},D,G)-\Phi(y_i=0,y_{-i},\mathcal{X},D,G)=U_i(y_i=1,y_{-i},\mathcal{X},D,G)-U_i(y_i=0,y_{-i},\mathcal{X},D,G).
    \end{equation}
    We have for any $i\in\mathcal{N}$,
    \begin{equation}
        \begin{split}
            &\quad\Phi(y_i=1,y_{-i},\mathcal{X},D,G)-\Phi(y_i=0,y_{-i},\mathcal{X},D,G)\\&=\alpha_i+ \sum_{j\in\mathcal{N}_i} \beta_{ij} y_j\\
            &=U_i(y_i=1,y_{-i},\mathcal{X},D,G)-U_i(y_i=0,y_{-i},\mathcal{X},D,G),
        \end{split}
    \end{equation}
    where the first equality holds since $j\in\mathcal{N}_i$ implies $i\in\mathcal{N}_j$, and $\beta_{ij}=\beta_{ji}$. Therefore, $\Phi$ is the potential of our interacted decision game.
 
\end{proof}

\begin{subsection}{Proof of Proposition \ref{prounimax}}\label{appenum}
\begin{proof}
    We first show that any maximizer is a solution of the first conditions by checking the following two conditions are satisfied.
    \begin{itemize}
        \item[1.] The objective function is continuous and differentiable in the interior.
        \item[2.] The boundary point cannot be a global optimum.
    \end{itemize}
Since our objective function is the sum of linear terms, quadratic functions and logarithmic functions, the first condition is trivially satisfied. 
To check the second condition, we need to verify that the derivative is positive at $\mu^Q_i=0$ and is negative at $\mu^Q_i=1$. 
The derivative is:
    \begin{equation}
        \begin{split}
            \frac{\partial }{\partial \mu^Q_i}\mathcal{A}(\mu^Q,\mathcal{X},D,G)&=\theta_0+\theta_1d_i+X_i'\theta_2+X_i'\theta_3d_i+A_N\sum_{j=1}^N\theta_4m_{ij}G_{ij}d_j\\&\quad+A_N\sum_{j=1}^Nm_{ij}G_{ij}(\theta_5+\theta_6d_id_j)\mu^Q_j-\log(\mu^Q_i)+\log(1-\mu^Q_i).
        \end{split}
    \end{equation}
    When $\mu^Q_i=0$,
    \begin{equation}
        -\log(\mu^Q_i)+\log(1-\mu^Q_i)=-\log(0)+\log(1)=+\infty.
    \end{equation}
    When $\mu^Q_i=1$,
    \begin{equation}
        -\log(\mu^Q_i)+\log(1-\mu^Q_i)=-\log(1)+\log(0)=-\infty.
    \end{equation}
    Since all of the elements in $\theta_0+\theta_1d_i+X_i'\theta_2+X_i'\theta_3d_i+A_N\sum_{j=1}^N\theta_4m_{ij}G_{ij}d_j+A_N\sum_{j=1}^Nm_{ij}G_{ij}(\theta_5+\theta_6d_id_j)\mu^Q_j$ are bounded,
    \begin{equation}
        \frac{\partial }{\partial \mu^Q_i}\mathcal{A}(\mu^Q,\mathcal{X},D,G)|_{\mu^Q_i=0}=+\infty,
    \end{equation}
    and 
    \begin{equation}
        \frac{\partial }{\partial \mu^Q_i}\mathcal{A}(\mu^Q,\mathcal{X},D,G)|_{\mu^Q_i=1}=-\infty.
    \end{equation}
    Away from the boundary, the objective function increases. 
    A global optimum therefore has to be in the interior, and by differentiability of the objective function, the first-order condition has to be satisfied at the optimum.

    Next, we apply the Banach fixed-point theorem to show that the solution of the first-order conditions is unique. 
    For this goal, we focus on the iteration procedure in Algorithm \ref{almu} and show it is a \textit{contraction mapping} for all $i\in \mathcal{N}$, for all $\{d_i\}_{i=1}^N\in \{0,1\}^{N}$, for all $\mathcal{X}\in\mathbb{R}^{N\times k}$, and for all $G\in\{0,1\}^{N\times N}$.    
    Recall the iteration in Algorithm \ref{almu}:
    \begin{equation}
           \Tilde{\mu}_i^{t+1}=\Lambda \Big[\theta_0+\theta_1 d_i+X_i'(\theta_2 +\theta_3d_i)+A_N\theta_4\sum\limits_{j\neq i}m_{ij}G_{ij}d_j+A_N\sum\limits_{j\neq i}m_{ij}G_{ij}(\theta_5+\theta_6d_id_j)\Tilde{\mu}_j^{t}\Big].
    \end{equation}
We denote this iteration process as $\{\Tilde{\mu}^{t}\}_{t=1}^\mathcal{T}$ and show the above mapping $T:[0,1]^{N} \rightarrow [0,1]^{N}$ is a contraction mapping. To prove the above iteration is a contraction mapping, we use $\ell_1$-distance. For any $t\geq 1$,
\begin{equation}
    d(T(\Tilde{\mu}^{t}),T(\Tilde{\mu}^{s}))=\sum_{i=1}^N\big\vert\Tilde{\mu}_i^{t+1}-\Tilde{\mu}_i^{s+1}\big\vert, 
\end{equation}

\begin{equation}
    d(\Tilde{\mu}^t,\Tilde{\mu}^{s}) = \sum_{i=1}^N\big\vert\Tilde{\mu}_i^t-\Tilde{\mu}_i^{s}\big\vert.
\end{equation}
First, since $\Lambda(\cdot)$ is a sigmoid function, its largest slope is 0.25, implying that 
\begin{equation}
    \begin{split}
        \left| \Tilde{\mu}^{t+1}_i-\Tilde{\mu}^{s+1}_i \right| &\leq 0.25 \left| A_N\sum_{j\neq i}(\theta_5+\theta_6d_id_j)m_{ij}G_{ij}(\Tilde{\mu}_j^t-\Tilde{\mu}_j^{s}) \right| \\
        &(\text{By Multivariate Mean Value Theorem})\\
        &\leq 0.25 A_N(\vert\theta_5\vert+\vert\theta_6\vert)\sum_{j\neq i}m_{ij}G_{ij}\vert\Tilde{\mu}_j^t-\Tilde{\mu}_j^{s}\big\vert.
    \end{split}
    \end{equation}
    Therefore,
    \begin{equation}
        \begin{split}
         \big\vert\Tilde{\mu}^{t+1}_i-\Tilde{\mu}^{s+1}_i\big\vert&\leq\frac{ A_N(\vert\theta_5\vert+\vert\theta_6\vert)}{4}\sum_{j\neq i}m_{ij}G_{ij}\vert\Tilde{\mu}_j^t-\Tilde{\mu}_j^{s}\big\vert.
        \end{split}
    \end{equation}
    Hence,
    \begin{equation}
        \begin{split}
          \sum_{i=1}^N\big\vert\Tilde{\mu}^{t+1}_i-\Tilde{\mu}^{s+1}_i\big\vert&\leq  \frac{ A_N(\vert\theta_5\vert+\vert\theta_6\vert)}{4}\sum_{i=1}^N\sum_{j\neq i}m_{ij}G_{ij}\vert\Tilde{\mu}_j^t-\Tilde{\mu}_j^{s}\big\vert\\
          &\leq \frac{ A_N(\vert\theta_5\vert+\vert\theta_6\vert)}{4}\sum_{i=1}^N\big\vert\Tilde{\mu}_i^t-\Tilde{\mu}_i^{s}\big\vert\max_{i\in\mathcal{N}}\sum_{j\neq i}m_{ij}G_{ij}.
        \end{split}
    \end{equation}
    where this inequality follows by the next inequality: for arbitrary $(b_j \in \mathbb{R}_+: j=1, \dots, N)$ and $(a_{ij} \in \mathbb{R}: 1 \leq i,j \leq N)$,
\begin{align}
\sum_{i=1}^N \sum_{j \neq i} a_{ij}b_j = \sum_{i=1}^N \sum_{j \neq i} a_{ji} b_i = \sum_{i=1}^N b_i \left( \sum_{j \neq i} a_{ji} \right) \leq \sum_{i=1} b_i \max_{1 \leq i \leq N} \left( \sum_{j \neq i} a_{ji} \right),
\end{align}
and the symmetry of $m_{ij}G_{ij}$.

    Therefore, under Assumption \ref{ass:boundspill}, $T$ is a contraction mapping. In addition, since $\Tilde{\mu}^t\in[0,1]$ for all $t\geq 1$, the metric space $(\Tilde{\mu},d)$ is a complete metric space. By Banach fixed-point theorem, the solution of $\Tilde{\mu} = T(\Tilde{\mu})$ is unique and $\{\Tilde{\mu}',\Tilde{\mu}^1,\Tilde{\mu}^2,...\}$ converges to a unique fixed point.
\end{proof}
\end{subsection}

\section{Theorem}\label{appendixC}
\subsection{Proof of Theorem \ref{theorem2}}\label{apptheo31}

We first state Theorem \ref{theorem16} that we use to prove Theorem \ref{theorem2}.
\begin{theorem}\citep[\S Theorem 1.6]{chatterjee2016nonlinear}\label{theorem16}
Suppose that $f : [0, 1]^N \rightarrow\mathbb{R}$ is twice continuously differentiable in $(0, 1)^N$, so that $f$ and all of its first- and second-order derivatives extend continuously to the boundary. 
Let $\lVert f\rVert$ denote the supremum norm of $f : [0, 1]^N \rightarrow\mathbb{R}$. 
For each $i$ and $j$, denote
\begin{equation}
    f_i:=\frac{\partial f}{\partial x_i},\quad f_{ij}:=\frac{\partial^2 f}{\partial x_i\partial x_j},
\end{equation}
and let
\begin{equation}
    a:=\lVert f\rVert, \quad b_i:=\lVert f_{i}\rVert,\quad c_{ij}:=\lVert f_{ij}\rVert.
\end{equation}
Given $\epsilon>0$, $\mathcal{M}(\epsilon)$ is a finite subset of $\mathbb{R}^N$ such that for any
$\Tilde{\mu}\in \{0, 1\}^N$, there exists $\eta=(\eta_1,...,\eta_N)\in\mathcal{M}(\epsilon)$ such that
\begin{equation}
 \sum_{i}\Big(\frac{\partial f(\Tilde{\mu})}{\partial\Tilde{\mu}_{i}}-\eta_i\Big)^2\leq N\epsilon^2.
\end{equation}
Let us define for any $\Tilde{\mu}=(\Tilde{\mu}_1,...,\Tilde{\mu}_N)\in[0,1]^N$,
\begin{equation}
    I(\Tilde{\mu}) = \sum_{i=1}^N\left[\Tilde{\mu}_i\log \Tilde{\mu}_i+(1-\Tilde{\mu}_i)\log(1-\Tilde{\mu}_i)\right].
\end{equation}
Let us define
\begin{equation}
    F\coloneqq\log\sum_{\mu\in\{0,1\}^N}\exp(f(\mu)).
\end{equation}
Then for an $\epsilon>0,$
\begin{equation}
    F\leq \sup_{\Tilde{\mu}\in[0,1]^N}(f(\Tilde{\mu})-I(\Tilde{\mu}))+\textit{complexity term} +\textit{smoothness term}.
\end{equation}
where
\begin{equation}
    \textit{complexity term} = \frac{1}{4}\Big(n\sum_{i=1}^N b_i^2\Big)^{1/2}\epsilon+3N\epsilon+\log\vert\mathcal{M}(\epsilon) \vert,
\end{equation}
and
\begin{equation}
 \begin{split}
     \textit{smoothness term} &= 4\Big(\sum_{i=1}^N(ac_{ii}+b_i^2)+\frac{1}{4}\sum_{i,j=1}^N (ac_{ij}^2+b_ib_jc_{ij}+4b_ic_{ij})\Big)^{1/2}\\ 
     &\quad+\frac{1}{4}\Big(\sum_{i=1}^N b_i^2\Big)^{1/2}\Big(\sum_{i=1}^N c_{ii}^2\Big)^{1/2}+3\sum_{i=1}^Nc_{ii}+\log 2.  
 \end{split} 
\end{equation}
\end{theorem}
\begin{proof}[Proof of Theorem \ref{theorem2}]
Define $f:[0,1]^{N}\rightarrow\mathbb{R}$ as:
\begin{equation}
    f(\Tilde{\mu})=\sum_{i}(\theta_0+\theta_1 d_i+X_i'(\theta_2 +\theta_3d_i)+A_N\theta_4\sum\limits_{j\neq i}m_{ij}G_{ij}d_j)\Tilde{\mu}_i+\frac{A_N}{2}\sum_{i}\sum_{j} (\theta_{5}+ \theta_{6}d_id_j)m_{ij}G_{ij}\Tilde{\mu}_i\Tilde{\mu}_j.
\end{equation}

Therefore,
\begin{equation}
   \begin{split}
        \lVert f\rVert &\leq \sum_{i}(\vert\theta_0 \vert +\vert \theta_1\vert+\vert X_i'\theta_2\vert+\vert \theta_3 X_i\vert)+A_N\sum_{i}\sum_{j}\vert\theta_4m_{ij}\vert G_{ij}+\frac{A_N}{2}\sum_{ i}\sum_{j} G_{ij}(\vert\theta_5m_{ij}\vert+\vert\theta_6m_{ij}\vert)\\
        &\leq N(\vert\theta_0 \vert+\vert\theta_1 \vert+\max_i\vert X_i'\theta_2 \vert+\max_i\vert X_i'\theta_3 \vert)+\overline{m}A_N(
        \vert \theta_4\vert+\vert \theta_5\vert+\vert \theta_6\vert)\sum_i\sum_{j}G_{ij}\\
        &\leq N(\vert\theta_0 \vert+\vert\theta_1 \vert+\max_i\vert X_i'\theta_2\vert+\max_i\vert X_i'\theta_3 \vert)+\overline{m}A_NN\widebar{N}(
        \vert \theta_4\vert+\vert \theta_5\vert+\vert \theta_6\vert)\\&\eqqcolon\Tilde{a}.
   \end{split}
\end{equation}
The partial derivative of $f(\Tilde{\mu})$ with respect to $\Tilde{\mu}_{i}$ is:
\begin{equation}\label{eq:pd}
    \frac{\partial f(\Tilde{\mu})}{\partial\Tilde{\mu}_{i}}=\theta_0+\theta_1+X_i'\theta_2+X_i'\theta_3d_i+A_N\sum_{j\neq i}\theta_4m_{ij}G_{ij}d_j+A_N\sum_{j\neq i}\theta_5m_{ij}G_{ij}\Tilde{\mu}_j+A_N\sum_{j\neq i}\theta_6m_{ij}G_{ij}d_id_j\Tilde{\mu}_j.
\end{equation}
Therefore,
\begin{equation}
    \begin{split}
       \lVert \nabla_i f(\Tilde{\mu})\rVert\leq  \vert \theta_0\vert+\vert\theta_1\vert +\max_i\vert X_i'\theta_2\vert+\max_i\vert X_i'\theta_3\vert+\overline{m}A_N\widebar{N}(
        \vert \theta_4\vert+\vert \theta_5\vert+\vert \theta_6\vert)\eqqcolon \Tilde{b}.
    \end{split}
\end{equation}
The second partial derivative with respect to $\Tilde{\mu}_{j}$ is:
\begin{equation}
    \frac{\partial^2 f(\Tilde{\mu})}{\partial\Tilde{\mu}_{i}\partial\Tilde{\mu}_{j}}=A_N\theta_5m_{ij}G_{ij}+A_N\theta_6m_{ij}G_{ij}d_id_j.
\end{equation}
Therefore, for all $j\neq i$,
\begin{equation}
    \lVert \nabla_i\nabla_j f(\Tilde{\mu})\rVert\leq \overline{m}A_N(\vert \theta_5\vert + \vert \theta_6\vert)G_{ij} =\Tilde{c}G_{ij}\eqqcolon \Tilde{c}_{ij},
\end{equation}
with the second derivative zero if $i=j$.
Next, we need to compute $\vert \mathcal{M}(\epsilon)\vert$, where $\mathcal{M}(\epsilon)$ is the finite subset of $\mathbb{R}^{N}$ such that for any $\Tilde{\mu}\in\{0,1\}^{N}$, there exists $\eta=(\eta_1,...,\eta_{N})\in \mathcal{M}(\epsilon)$ such that
\begin{equation}\label{eq:mathmcon}
 \sum_{i}\Big(\frac{\partial f(\Tilde{\mu})}{\partial\Tilde{\mu}_{i}}-\eta_i\Big)^2\leq N\epsilon^2.
\end{equation}
Recalling Eq.\ref{eq:pd} and defining $T_1$ and $T_2$ as
\begin{equation}
    \begin{split}
      T_1(\Tilde{\mu})\coloneqq\frac{A_N\theta_5}{2}\sum_{i}\sum_{j}m_{ij}G_{ij}\Tilde{\mu}_i\Tilde{\mu}_j,\quad T_2(\Tilde{\mu})\coloneqq\frac{A_N\theta_6}{2}\sum_{i}\sum_{j}m_{ij}G_{ij}d_id_j\Tilde{\mu}_i\Tilde{\mu}_j,
    \end{split}
\end{equation}
we state the partial derivative of $f(\Tilde{\mu})$ as
\begin{equation}\label{Eq:125}
    \begin{split}
        \frac{\partial f(\Tilde{\mu})}{\partial\Tilde{\mu}_{i}}&=\theta_0+\theta_1+X_i'\theta_2+X_i'\theta_3d_i+\theta_4A_N\sum_{j}m_{ij}G_{ij}d_j+\frac{\partial T_1(\Tilde{\mu})}{\partial \Tilde{\mu}_{i}}+\frac{\partial T_2(\Tilde{\mu})}{\partial \Tilde{\mu}_{i}}.
    \end{split}
\end{equation}
To construct a $\mathcal{M}(\epsilon)$ that satisfies Eq.\ref{eq:mathmcon}, we define $\mathcal{M}(\epsilon)$ as:
\begin{equation}
    \begin{split}
        \mathcal{M}(\epsilon)&:=\Big\{\theta_0+\theta_1+X_i'\theta_2+X_i'\theta_3d_i+A_N\theta_4\sum_{j}m_{ij}G_{ij}d_j+\ell_1+\ell_2\\&\quad:\ell_1\in\mathcal{M}_1\left(\frac{\epsilon}{\sqrt{2}}\right),\ell_2\in\mathcal{M}_2\left(\frac{\epsilon}{\sqrt{2}}\right),i\in\mathcal{N}\Big\}.
    \end{split}
\end{equation}
We now first need to construct a set $\mathcal{M}_1(\epsilon)$, which is a finite subset of $\mathbb{R}^{N}$, such that for any $\tilde{\mu} \in \{0,1\}^{N}$, there exists a vector $\lambda = (\lambda_1,\ldots,\lambda_N) \in \mathcal{M}_1(\epsilon)$ satisfying
\begin{equation}\label{Eq:M1e}
 \sum_{i}\Big(\frac{\partial T_1(\Tilde{\mu})}{\partial\Tilde{\mu}_{i}}-\lambda_i\Big)^2\leq N\epsilon^2.
\end{equation}
We then need to construct a set $\mathcal{M}_2(\epsilon)$, which is a finite subset of $\mathbb{R}^{N}$, such that for any $\tilde{\mu} \in \{0,1\}^{N}$, there exists a vector $\vartheta=(\vartheta_1,...,\vartheta_{N})\in \mathcal{M}_2(\epsilon)$ satisfying
\begin{equation}\label{Eq:M2e}
 \sum_{i}\Big(\frac{\partial T_2(\Tilde{\mu})}{\partial\Tilde{\mu}_{i}}-\vartheta_i\Big)^2\leq N\epsilon^2.
\end{equation}
To do so, we define $\lambda_i$ and $\vartheta_i$ for all $i\in\mathcal{N}$ as
\begin{equation}
    \lambda_i = A_N\theta_5\sum_{j}m_{ij}G_{ij}\Tilde{y}_j,\quad \vartheta_i = A_N\theta_6\sum_{j}m_{ij}G_{ij}d_id_j\Tilde{v}_j,
\end{equation}
for some $\Tilde{y}_j\in\{0,1\}$ and $\Tilde{v}_j\in\{0,1\}$. Then,
\begin{equation}
    \begin{split}
        \Big(\frac{\partial T_1(\Tilde{\mu})}{\partial\Tilde{\mu}_{i}}-\lambda_i\Big)^2&=\Big(A_N\theta_5\sum_{j}m_{ij}G_{ij}(\Tilde{\mu}_j-\Tilde{y}_j)\Big)^2\\
         &\leq \Big(A_N\theta_5\sum_{j}m_{ij}^2G_{ij}\Big)\Big(A_N\theta_5\sum_{j}(\Tilde{\mu}_j-\Tilde{y}_j)^2\Big)\\
        &(\text{By the Cauchy–Schwarz inequality})\\
        &\leq A_N^2\theta_5^2N\sum_{j}(\Tilde{\mu}_j-\Tilde{y}_j)^2\max_{i,j}m_{ij}^2.
    \end{split}
\end{equation}
By Lemma \ref{lemmatau}, there always exists $\Tilde{y}_j\in\{0,1\}$ such that for any $\tau\in\mathbb{R}_{+}$
\begin{equation}
    \sum_{j}(\Tilde{\mu}_j-\Tilde{y}_j)^2\leq \tau^2.
\end{equation}
Therefore, by choosing $\tau_1=\sqrt{\frac{\epsilon^2}{2A_N^2\theta_5^2N\max_{ij}m_{ij}^2}}$, we have
\begin{equation}
   \begin{split}
        \Big(\frac{\partial T_1(\Tilde{\mu})}{\partial\Tilde{\mu}_{i}}-\lambda_i\Big)^2\leq A_N^2\theta_5^2N\tau_1^2\max_{i,j}m_{ij}^2
        =\frac{\epsilon^2}{2}
   \end{split}
\end{equation}
Applying the same argument to $\Tilde{v}_j\in\{0,1\}$,
\begin{equation}
    \begin{split}
        \Big(\frac{\partial T_2(\Tilde{\mu})}{\partial\Tilde{\mu}_{i}}-\vartheta_i\Big)^2&=\Big(A_N\theta_6\sum_{j}m_{ij}G_{ij}d_id_j(\Tilde{\mu}_j-\Tilde{v}_j)\Big)^2\\
         &\leq \Big(A_N\theta_6\sum_{j}m_{ij}^2G_{ij}d_id_j\Big)\Big(A_N\theta_6\sum_{j}(\Tilde{\mu}_j-\Tilde{v}_j)^2\Big)\\
        &(\text{By the Cauchy–Schwarz inequality})\\
        &\leq A_N^2\theta_6^2N\sum_{j}(\Tilde{\mu}_j-\Tilde{v}_j)^2\max_{i,j}m_{ij}^2.
    \end{split}
\end{equation}
Again, by Lemma \ref{lemmatau}, there always exists $\Tilde{v}_j\in\{0,1\}$ such that for any $\tau\in\mathbb{R}_{+}$, $\sum_{j}(\Tilde{\mu}_j-\Tilde{v}_j)^2\leq \tau^2$. By choosing $\tau_2=\sqrt{\frac{\epsilon^2}{2A_N^2\theta_6^2N\max_{ij}m_{ij}^2}}$
\begin{equation}
    \Big(\frac{\partial T_2(\Tilde{\mu})}{\partial\Tilde{\mu}_{i}}-\vartheta_i\Big)^2\leq A_N^2\theta_6^2 N\tau_2^2\max_{i,j}m_{ij}^2=\frac{\epsilon^2}{2}.
\end{equation}
Therefore, 
\begin{equation}
    \vert\mathcal{M}(\epsilon)\vert\leq N\cdot \frac{N(N+1)}{2}\cdot\Big\vert\mathcal{M}_1\left(\frac{\epsilon}{\sqrt{2}}\right)\Big\vert\cdot\Big\vert\mathcal{M}_2\left(\frac{\epsilon}{\sqrt{2}}\right)\Big\vert\leq 2^{2N-1}N^2(N+1).
\end{equation}
We now apply Theorem \ref{theorem16}, choosing $\epsilon=N^{-1}$, and noting the current choice of $f$ leads to
\begin{equation}
    \begin{split}
        \mathbb{K}(Q^*\lVert P)&
        \leq \frac{1}{4}(N\sum_{i}b_i^2)^{\frac{1}{2}}N^{-1}+3+\log(2^{2N-1}(N^3+N^2))\\&\quad+4\Big(\sum_{i}b_i^2+\frac{1}{4}\sum_{i,j}(ac_{ij}^2+b_ib_jc_{ij}+4b_ic_{ij})\Big)^{\frac{1}{2}}+\log2\\
        &\leq \frac{\Tilde{b}}{4}+3+2N\log 2+\log(N^3+N^2)+4\Big(\Tilde{b}^2N+\frac{1}{4}\sum_{i,j}(\Tilde{a}\Tilde{c}^2G_{ij}+\Tilde{b}^2\Tilde{c}G_{ij}+4\Tilde{b}\Tilde{c}G_{ij})\Big)^{\frac{1}{2}}\\
        &\leq \frac{\Tilde{b}}{4}+3+2N\log 2+\log(N^3+N^2)+4\Big(\Tilde{b}^2N+\frac{1}{4}\widebar{N}N(\Tilde{a}\Tilde{c}^2+\Tilde{b}^2\Tilde{c}+4\Tilde{b}\Tilde{c})\Big)^{\frac{1}{2}}\\
        &=C_1A_N\widebar{N}+C_2N\\&+\sqrt{(C_3A_N\widebar{N}+C_4A_N^2\widebar{N}^2)N+(C_5A_N+C_6A_N^2\widebar{N}+C_7A_N^3\widebar{N}^2+C_8A_N^2+C_9\widebar{N}A_N^3)N^2}+o(N)\\
        &=C_1A_N\widebar{N}+C_2N+\mathcal{O}\left(\sqrt{A_N^2\widebar{N}^2N}\right)+\mathcal{O}\left(\sqrt{A_N^3\widebar{N}^2N^2}\right)+\mathcal{O}\left(\sqrt{A_N^3\widebar{N}N^2}\right)+o(N),
    \end{split}
\end{equation}
\normalsize
where $o(N)$ collects those elements that are constant or that grow at a slower rate than $N$, and
\begin{equation}
    C_1 = \frac{1}{4}\overline{m}(\vert \theta_4\vert+\vert \theta_5\vert+\vert \theta_6\vert),
\end{equation}
\begin{equation}
    C_2 = 2\log2,
\end{equation}
\begin{equation}
    C_3=32\overline{m}(\vert \theta_0\vert+\vert\theta_1\vert +\max_i\vert X_i'\theta_2\vert+\max_i\vert X_i'\theta_3\vert)(\vert \theta_4\vert+\vert \theta_5\vert+\vert \theta_6\vert)
\end{equation}
\begin{equation}
    C_4=16\overline{m}^2(\vert \theta_4\vert+\vert \theta_5\vert+\vert \theta_6\vert)^2
\end{equation}
\begin{equation}
    C_5=4\overline{m}(\vert \theta_0\vert+\vert\theta_1\vert +\max_i\vert X_i'\theta_2\vert+\max_i\vert X_i'\theta_3\vert+4)(\vert \theta_0\vert+\vert\theta_1\vert +\max_i\vert X_i'\theta_2\vert+\max_i\vert X_i'\theta_3\vert)(\vert\theta_5 \vert+\vert\theta_6 \vert).
\end{equation}
\begin{equation}
    C_6=8\overline{m}^2[(\vert \theta_0\vert+\vert\theta_1\vert +\max_i\vert X_i'\theta_2\vert+\max_i\vert X_i'\theta_3\vert)+2](\vert \theta_4\vert+\vert \theta_5\vert+\vert \theta_6\vert)(\vert\theta_5 \vert+\vert\theta_6 \vert),
    \end{equation}
    \begin{equation}
    C_7=4\overline{m}^3(
        \vert \theta_4\vert+\vert \theta_5\vert+\vert \theta_6\vert)^2(\vert\theta_5 \vert+\vert\theta_6 \vert),
    \end{equation}
    \begin{equation}
    C_8=4\overline{m}^2(\vert\theta_0 \vert+\vert\theta_1 \vert+\max_i\vert X_i'\theta_2 \vert+\max_i\vert X_i'\theta_3 \vert)(\vert\theta_5 \vert+\vert\theta_6 \vert)^2,
    \end{equation}
    \begin{equation}
    C_9=4\overline{m}^3(
        \vert \theta_4\vert+\vert \theta_5\vert+\vert \theta_6\vert)(\vert\theta_5 \vert+\vert\theta_6 \vert)^2.
    \end{equation}
\end{proof}

\subsection{Proof of Theorem \ref{theoremgreedy}}\label{appentheorem32}
\begin{proof}

We need to apply Theorem \ref{theorembian}, which states

\begin{theorem}\citep[\S Theorem 1]{bian2017guarantees}\label{theorembian}
Let $F(\cdot)$ be a non-negative nondecreasing set function with submodularity ratio $\gamma\in[0, 1]$ and curvature
$\xi\in [0, 1]$. The greedy algorithm enjoys the following approximation guarantee for solving the maximization problem with cardinality constraint:
\begin{equation}
    F(D_G)\geq \frac{1}{\xi}(1-e^{-\xi\gamma}) F(D^*),
\end{equation}
where $D_G$ is the result of the greedy algorithm and $D^*$ is the optimal solution.
\end{theorem}
The definitions of submodularity, the submodularity ratio, and the curvature of a set function $f$ are as follows.
\begin{definition}{(\textbf{Submodularity})}: A set function is a submodular function if:
\begin{equation}
    \sum_{k\in R\setminus S} [f(S\cup \{k\})-f(S)]\geq  f(S\cup R)-f(S), \quad\forall S,R\subseteq\mathcal{N}.
\end{equation}
\end{definition}
\begin{definition}
    (\textbf{Submodularity Ratio}) The submodularity ratio of a non-negative set function $f(\cdot)$ is the largest $\gamma$ such that
\begin{equation}
    \sum_{k\in R\setminus S} [f(S\cup \{k\})-f(S)]\geq \gamma [f(S\cup R)-f(S)], \quad\forall S,R\subseteq\mathcal{N}.
\end{equation}
\end{definition}

\begin{definition}
    (\textbf{Curvature}) The curvature of a non-negative set function $f(\cdot)$ is the smallest value of $\xi$ such that
\begin{equation}
f(R\cup \{k\})-f(R)\geq (1-\xi)[f(S\cup \{k\})-f(S)],\quad   \forall S\subseteq R\subseteq \mathcal{N}, \forall k\in \mathcal{N}\setminus R.
\end{equation}
\end{definition}

The submodularity of a set function is analogous to concavity of a real function and implies that the function has diminishing returns. 
The marginal increase in the probability of choosing action $1$ decreases with the number of treated units. 
The submodularity ratio captures how much greater the probability of choosing action $1$ is from providing treatment to a group of units versus the combined benefit of treating each unit individually. 
Curvature can be interpreted as how close a set function is to being additive.

We associate the set function $f(\cdot)$ in the above definitions with the variationally approximated welfare $\Tilde{W}(\cdot)$, which we view as a real-valued mapping of treatment allocation sets $\mathcal{D} \subset \mathcal{N}$ (i.e., $\mathcal{D}=\{i\in\mathcal{N}:d_i=1\}$):

\begin{equation}\label{eqset}
            \begin{split}
                \Tilde{W}(\mathcal{D})&= \sum_{i\in \mathcal{D}}\Lambda\big[\theta_0+\theta_1 +X_i'(\theta_2+\theta_3) +A_N\theta_5 \sum_{\substack{j\neq i\\j\in\mathcal{N}}}m_{ij}G_{ij}\Tilde{\mu}_j+A_N\sum_{\substack{j\neq i\\j\in \mathcal{D}}}m_{ij}G_{ij}(\theta_4+\theta_6\Tilde{\mu}_j)\big]\\&\quad+\sum_{k\in\mathcal{N}\setminus\mathcal{D}}\Lambda\big[\theta_0+X_k'\theta_2+A_N\theta_4\sum_{\ell\in\mathcal{D}}m_{k\ell}G_{k\ell}+A_N\theta_5\sum_{\substack{\ell\neq k\\\ell\in\mathcal{N}}}m_{k\ell}G_{k\ell}\Tilde{\mu}_{\ell}\big].
            \end{split}
        \end{equation}  
We characterize the submodularity ratio and curvature of $\tilde{W}(\cdot)$ to obtain an analytical performance guarantee for our greedy algorithm.

Under Assumption \ref{assptse}(i), set function $\Tilde{W}(\cdot)$ is non-decreasing, its curvature $\xi$ and its submodularity ratio $\gamma$ must belong to $[0,1]$ \citep{bian2017guarantees}. 
Having $\xi \in [0,1]$ and $\gamma \in [0,1]$ is not, however, enough to attain a nontrivial performance guarantee. 
For instance, if $\gamma=0$, the lower bound in Theorem \ref{theoremgreedy} equals 0, which is a trivial lower bound; 
if $\xi=0$, then the lower bound equals $\gamma$, which could be $0$. when $\gamma = 1$ and $\xi = 0$, this represents the most favorable scenario for our performance guarantee. A lower value of $\gamma$ implies that the objective function is less likely to exhibit the submodularity property, meaning it no longer satisfies diminishing returns. In such cases, the greedy algorithm becomes less effective because selecting the unit with the highest marginal gain is no longer well justified. Similarly, when the curvature $\xi$ is high, the function becomes less additive, indicating that selecting only one unit at each iteration of the greedy algorithm is suboptimal. This provides an intuition for why the problem becomes more challenging in these settings. Assumption \ref{assptse} (ii) gives a sufficient condition to bound the submodularity ratio and curvature away from 0 and 1.

To summarize, the first statement in Theorem \ref{theoremgreedy} directly follows Theorem \ref{theorembian} with $F(D) = \Tilde{W}(D)$. Let $C_{greedy}=1-\frac{1}{\xi}(1-e^{-\xi\gamma})$, Lemma \ref{lemmabound} guarantees $C_{greedy}>0$.  
\end{proof}

\end{document}